%% file: LCKmain.tex
\documentclass[11pt,letterpaper]{article}
\usepackage[top=1in, bottom=1in, left=1in, right=1in]{geometry}  
\input{ycPre.tex}
\newcommand{\mybib}{LCKrefs}
\newcommand{\mybibsty}{agsm}
\renewcommand{\baselinestretch}{1.2} 
\graphicspath{ {images/} }

\usepackage{graphicx}
\usepackage[table,xcdraw]{xcolor}
\usepackage{arydshln}  
\begin{document}

\title{Multivariate Spatial-temporal Prediction on Latent Low-dimensional Functional Structure with Non-stationarity}
\author{Yi Chen \hspace{2ex} Qiwei Yao \hspace{2ex} Rong Chen}
\date{\today}
\maketitle

\begin{abstract}
Multivariate spatio-temporal data arise more and more frequently in a wide range of applications; however, there are relatively few general statistical methods that can readily use that incorporate spatial, temporal and variable dependencies simultaneously. In this paper, we propose a new approach to represent non-parametrically the linear dependence structure of a multivariate spatio-temporal process in terms of latent common factors. The matrix structure of observations from the multivariate spatio-temporal process is well reserved through the matrix factor model configuration. The spatial loading functions are estimated non-parametrically by sieve approximation and the variable loading matrix is estimated via an eigen-analysis of a symmetric non-negative definite matrix. Though factor decomposition along the space mode is similar to the low-rank approximation methods in spatial statistics, the fundamental difference is that the low-dimensional structure is completely unknown in our setting. Additionally, our method accommodates non-stationarity over space. The estimated loading functions facilitate spatial prediction. For temporal forecasting, we preserve the matrix structure of observations at each time point by utilizing the matrix autoregressive model of order one MAR(1). Asymptotic properties of the proposed methods are established. Performance of the proposed method is investigated on both synthetic and real datasets.

\vspace{1em}

\end{abstract}

\section{Introduction}


The increasing availability of multivariate data referenced over geographic regions and time in various applications has created unique opportunities and challenges for those practitioners seeking to capitalize on their full utility. For example, United States Environmental Protection Agency publishes daily from more than 20,000 monitoring stations a collection of environmental and meteorological measurements such as temperature, pressure, wind speed and direction and various pollutants. Such data naturally constitute a tensor (multi-dimensional array) with three modes (dimensions) representing space, time and variates, respectively. Simultaneously modeling the dependencies between different variates, regions, and times is of great potential to reduce dimensions, produce more accurate estimation and prediction and further provide a deeper understanding of the real world phenomenon. At the same time, methodological issues arise because these data exhibit complex multivariate spatio-temporal covariances that may involve non-stationarity and potential dependencies between spatial locations, time points and different processes. Traditionally, researchers mainly restrict their analysis to only two dimensions while fixing the third: time series analysis applied to a slice of such data at one location focus on temporal modeling and prediction (\cite{box2015time, brockwell2013time, tsay2013multivariate, fanyao2005nonlinear}); spatial statistical models for a slice of such data at one time point address spatial dependence and prediction over unobserved locations (\cite{cressie2015statisticsS}); and univariate spatio-temporal statistics concentrate on only one variable observed over space and time (\cite{cressie2015statisticsST}). 


Since physical processes rarely occur in isolation but rather influence and interact with one another, multivariate spatio-temporal models are increasingly in demand because the dependencies between multiple variables, locations and times can provide valuable information for understanding real world phenomenons. Various multivariate spatio-temporal conditional autoregressive models have been proposed by \cite{carlin2003hierarchical, congdon2004multivariate,pettitt2002conditional,zhu2005generalized,daniels2006conditionally,tzala2008bayesian}, among others. However, these methodologies cannot efficiently model high-dimensional data sets. Additionally, these approaches impose separability and various independence assumptions, which are not appropriate for many settings, as these models fails to capture important interactions and dependencies between different variables, regions, and times (\cite{stein2005space}). \cite{bradley2015multivariate} introduced a multivariate spatio-temporal mixed effects model to analyze high-dimensional multivariate data sets that vary over different geographic regions and time points. They adopt a reduced rank spatial structure (\cite{wikle2010low-rank}) and model temporal behavior via vector autoregressive components. However, their method only applies to low-dimensional multivariate observations because they model each variable separately. In addition, they assume the random effect term is common across all processes which is unrealistic especially in the case with a large number of variables. 

In this paper, we propose a new class of multivariate spatio-temporal models that model spatial, temporal and variate dependence simultaneously. The proposed model builds upon the matrix factor models proposed in \cite{wang2017factor}, while further incorporating the functional structure of the spatial process and dynamics of the latent matrix factor. The spatial dependence is model by the spatial loading functions, the variable dependence is modeled by the variable loading matrix, while the temporal dependence is modeled by the latent factors of first-order autoregressive matrix time series. 

Some spatial-factor-analysis models that capture spatial dependence through factor processes have been developed in the literature. \cite{lopes2008spatial} considers univariate observations but uses factor analysis to reduce (identify) clusters/groups of locations/regions whose temporal behavior is primarily described by a potentially small set of common dynamic latent factors. Also working with the univariate case, \cite{cressie2008fixed} successfully reduces the computational cost of kriging by using a flexible family of non-stationary covariance functions constructed from low rank basis functions. See also \cite{wikle2010low-rank}. For multivariate spatial data, \cite{cook1994some} introduced the concept of a spatially shifted factor and a single-factor shifted-lag model and \cite{majure1997dynamic} discussed graphical methods for identifying shifts. Following the ideas of multiple-lag dynamic factor models that generalize static factor models in the time series setting, \cite{christensen2001generalized, christensen2002latent, christensen2003modeling} extended the shifted-lag model to a generalized shifted-factor model by adding multiple shifted-lags and developed a systematic statistical estimation, inference, and prediction procedure. The assumption that spatial processes are second-order stationary is required for the moment-based estimation procedure and the theoretical development. Our modeling of the spatial dependence though latent factor processes is different from the aforementioned methods in that we impose no assumptions about the stationarity over space, nor the distribution of data, nor the form of spatial covariance functions. The idea is similar to that of \cite{huang2016krigings}, however we aim at estimating the spatial loading functions instead of the loading matrix and kriging at unsampled location is based on the loading function. In addition, future forecasting in our model reserves the matrix formation of the observation and temporal dependence through the matrix auto-regression of order one.  

The remainder of the article is outlined as follows. Section 2 introduces the model settings. Section 3 discusses estimation procedures for loading matrix and loading functions. Section 4 discuss the procedures for kriging and forecasting over space and time, respectively. Section 5 presents the asymptotic properties of the estimators. Section 6 illustrates the proposed model and estimation scheme on a synthetic dataset; And finally Section 7 applies the proposed method to a real dataset. Technique proofs are relegated to the Appendix.  

\section{The Model} \label{sec:model}

Consider a $p$-dimension multivariate spatio-temporal process $\by_t(\bs)=(y_{t,1}(\bs), \ldots, y_{t,p}(\bs))'$
\begin{equation} \label{eqn:stvp}
\by_t(\bs) = \bC'(s) \bz_t(\bs) + \bxi_t(\bs) + \bepsilon_t(\bs), \quad t = 0, \pm 1, \pm 2, \cdots, \; \bs \in \calS \subset \calR^2, 
\end{equation}
where $\bz_t(\bs)$ is an $m \times 1$ observable covariate vector, $\bC(s)$ is a $m \times p$ unknown parameter matrix, the additive error vector $\bepsilon_t(\bs)$ is unobservable and constitutes the nugget effect over space in the sense that
\begin{equation}
\E{\bepsilon_t(\bs)} = \vct{0}, \quad \Var{\bepsilon_t(\bs)} = \bSigma_{\bepsilon}(\bs), \quad \Cov{\bepsilon_{t_{1}}(\bu), \bepsilon_{t_{2}}(\bv)} = \mtx{0} \; \forall \; (t_{1}, \bu) \ne (t_{2}, \bv), 
\end{equation} 
$\bxi_t(\bs)$ is a $p$-dimension latent spatio-temporal vector process satisfying the condtions
\begin{equation}
\E{\bxi_t(\bs)} = \vct{0}, \quad \Cov{\bxi_{t_{1}}(\bu), \bxi_{t_{2}}(\bv)} = \bSigma_{| t_{1} - t_{2} |}(\bu, \bv).
\end{equation}

Under the above condtions, $\by_t(\bs) - \bC'(s) \bz_t(\bs)$ is seond order stationary in time $t$,
\begin{eqnarray*}
&\E{\by_t(\bs) - \bC'(s) \bz_t(\bs)} = \vct{0}, \\
&\Cov{\by_{t_1}(\bu) - \bC'(\bu) \bz_{t_1}(\bu), \by_{t_2}(\bv) - \bC'(\bv) \bz_{t_2}(\bv)} = \bSigma_{| t_{1} - t_{2} |}(\bu, \bv) + \bSigma_{\bepsilon}(\bu) \cdot \I{(t_{1}, \bu) = (t_{2}, \bv)}.
\end{eqnarray*}
Finally, we assume that $\bSigma_t(\bu, \bv)$ is continuous in $\bu$ and $\bv$. Note that model (\ref{eqn:stvp}) does not impose any stationary conditions over space, though it requires that $\by_t(\bs)$ is second order stationary in time $t$.  

We assume that the latent spatial-temporal vector process are driven by a lower-dimention latent spatial-temporal factor process, that is 
\begin{equation} \label{eqn:xi_fac}
\bxi_t(\bs) = \bB \bff_t(\bs),
\end{equation}
where $\bff_t(\bs)$ is the $r$-dimensional latent factor process ($r \ll p$) and $\bB$ is the $p \times r$ loading matrix.

Further, we assume that the latent $r \times 1$ factor process $\bff_t(\bs)$ admits a finite functional structure, 
\begin{equation} \label{eqn:fac_finstr}
\bff_t(\bs) = \sum_{j=1}^{d} a_j(\bs) \bx_{tj} ,
\end{equation}
where $a_1(\cdot), \cdots, a_d(\cdot)$ are deterministic and linear independent functions (i.e. none of them can be written as a linear combination of the others) in the Hilbert space $L_2(\calS)$, and $\bx_{tj} = \left( \bx_{tj,1}, \ldots, \bx_{tj,r} \right)$ is a $r \times 1$ random vector. Combining (\ref{eqn:xi_fac}) and (\ref{eqn:fac_finstr}), we have 
\begin{equation} \label{eqn:xi_fac_finstr}
\bxi_t(\bs) = \bB \sum_{j=1}^{d} a_j(\bs) \bx_{tj} = \bB \bX'_t \ba(\bs),
\end{equation}
where $\bX_t=\left( \bx_{t1}, \cdots, \bx_{td} \right)'$ and $\ba(\bs) = \left( a_1(\bs), \cdots, a_d(\bs) \right)'$. 

Stacking $\bxi_t(\bs)$ from $n$ locations $\bs_1, \ldots, \bs_n$ together as rows, we have a $n \times p$ matrix of $p$ signals from $n$ locations $\bXi_t = \left( \bxi_t(\bs_1), \cdots, \bxi_t(\bs_n) \right)'$. It follows from (\ref{eqn:xi_fac_finstr}) that
\begin{equation} \label{eqn:xi_fac_finstr_mat}
\bXi_t = \bA \bX_t \bB',
\end{equation} 
where  $\bA = [A_{ij}] = [a_j(\bs_i)]$, $i=1,\ldots,n$ and $j=1,\ldots,d$.

Obviously $a_1(\cdot), \cdots, a_d(\cdot)$ are not uniquely defined by (\ref{eqn:fac_finstr}) and $\bB$ is not uniquely defined by (\ref{eqn:xi_fac}). We assume that $a_1(\cdot), \cdots, a_d(\cdot)$ are orthonormal in the sense that $\langle a_j \, , a_k \rangle = \I{j=k}$ and $\bB'\bB = \bI_r$. Thus, the kernel reproducing Hilbert space (KRHS) spanned by $a_1(\cdot), \cdots, a_d(\cdot)$ and the vector space spanned by columns of $\bB$ (i.e. $\calM(\bB)$) are uniquely defined. We estimate the KRHS and $\calM(\bB)$ in this artical.

\section{Estimation} \label{sec:est}

Let $\{ \left( \by_t(\bs_i), \bz_t(\bs_i) \right), \quad i=1, \ldots, n, \quad t=1, \ldots, T  \}$ be the available observations over space and time, where $\by_t(\bs_i)$ is a vector of $p$ variables and $\bz_t(\bs_i)$ is a vector of $m$ covariates observed at location $\bs_i$ at time $t$. In this article, we restrict attention to the isotopic case where all variables have been measured at the same sample locations $\bs_i$, $i=1, \ldots, n$. 

To simplify the notation, we first consider a special case where $\bC(\bs) \equiv \bzero$ in (\ref{eqn:stvp}). Now the observations are from the process
\begin{equation} \label{eqn:stvp_c0}
\by_t(\bs) = \bxi_t(\bs) + \bepsilon_t(\bs) = \bB \bX'_t \ba(\bs) + \bepsilon_t(\bs).  
\end{equation}
Stacking $\by_t(\bs_i)$, $i = 1, \ldots, n$ together as rows, we have 
\begin{equation} \label{eqn:stvp_c0_mat}
\bY_t =  \bXi_t + \bE_t  = \bA \bX_t \bB' + \bE_t,
\end{equation} 
where $\bY_t = \left( \by_t(\bs_1), \cdots, \by_t(\bs_n) \right)$ and $\bE_t = \left( \bepsilon_t(\bs_1), \cdots, \bepsilon_t(\bs_n) \right)'$.
 
\subsection{Estimation for the Partitioned Spatial Loading Matrices $\bA_1$ and $\bA_2$}  \label{subsec:est:A1A2}

To exclude nugget effect in our estimation, we divide $n$ locations $\bs_1, \ldots, \bs_n$ into two sets $\calS_1$ and $\calS_2$ with $n_1$ and $n_2$ elements respectively. Let $\bY_{lt}$ be a matrix consisting of $\by_t(\bs)$, $\bs \in \calS_l$, $l=1,2$ as rows. Then $\bY_{1t}$ and $\bY_{2t}$ are two matrices of dimention $n_1 \times p$ and $n_2 \times p$ respectively. It follows from (\ref{eqn:stvp_c0}) that\\
\begin{equation} \label{eqn:stvp_mat_div}
\bY_{1t} =  \bXi_{1t} + \bE_{1t}  = \bA_1 \bX_t \bB' + \bE_{1t}, \qquad \bY_{2t} =  \bXi_{2t} + \bE_{2t}  = \bA_2 \bX_t \bB' + \bE_{2t},
\end{equation}
where $\bA_l$ is a $n_l \times d$ matrix, its rows are $\left( a_1(\bs), \ldots, a_d(\bs) \right)$ at diffent locations $\bs \in \calS_l$ and $\bE_{t,l}$ consists of $\bepsilon_t(\bs)$ as rows with $\bs \in \calS_l$, $l=1,2$.

For model identification, we assume $\bA'_1\bA_1 = \bI_d$ and $\bA'_2\bA_2 = \bI_d$, which however implies that $\bX_t$ in the second equation in (\ref{eqn:stvp_mat_div}) will be different from that in the first eqaution. Thus, we may rewrite (\ref{eqn:stvp_mat_div}) as
\begin{equation} \label{eqn:stvp_mat_div_1}
\bY_{1t} =  \bXi_{1t} + \bE_{1t}  = \bA_1 \bX_t \bB' + \bE_{1t}, \qquad \bY_{2t} =  \bXi_{2t} + \bE_{2t}  = \bA_2 \bX^*_t \bB' + \bE_{2t},
\end{equation}
where $\bX^*_t = \bQ \bX^*_t$ and $\bQ$ is an invertible $d \times d$ matrix. Under this assumption, $\calM(\bA_1)$ and $\calM(\bA_2)$, which are the column spaces of $\bA_1$ and $\bA_2$, are uniquely defined. 

Let $Y_{lt, \cdot j}$ be the $j$-th column of $\bY_{lt}$, $E_{lt, \cdot j}$ be the $j$-th column of $\bE_{lt}$ and $B_{j \cdot}$ be the $j$-th row of $\bB$, $l=1,2$ and $j = 1, \ldots, p$. Define spatial-cross-covariance matrix between the $i$-th and $j$-th variables as
\begin{eqnarray} \label{eqn:Omega_ij}
&\bOmega_{A,ij} & = \Cov{Y_{1t, \cdot i}, Y_{2t, \cdot j}} \nonumber \\
& & = \Cov{ \bA_1 \bX_t B_{i \cdot} + E_{1t,i}, \bA_2 \bX^*_t B_{j \cdot} + E_{2t,j }} \nonumber \\
& & = \bA_1 \Cov{ \bX_t B_{i \cdot}, \bX^*_t B_{j \cdot} } \bA_2 
\end{eqnarray} 

When $n \ll d$, it is reasonable to assume that $\rank{\bOmega_{A,ij}} = \rank{ \Cov{ \bX_t B_{i \cdot}, \bX^*_t B_{j \cdot} } } = d$. 

Define 
\begin{eqnarray} 
& \bM_{A_1} & = \sum_{i=1}^{p} \sum_{j=1}^{p} \bOmega_{A,ij} \bOmega'_{A,ij} \nonumber \\
& & = \bA_1 \left \{ \sum_{i=1}^{p} \sum_{j=1}^{p} \Cov{ \bX_t B_{i \cdot}, \bX^*_t B_{j \cdot} } \Cov{ \bX^*_t B_{j \cdot}, \bX_t B_{i \cdot} } \right \}  \bA'_1,  \label{eqn:M1} \\
& \bM_{A_2} & = \sum_{i=1}^{p} \sum_{j=1}^{p} \bOmega'_{A,ij} \bOmega_{A,ij} \nonumber \\
& & = \bA_2 \left \{ \sum_{i=1}^{p} \sum_{j=1}^{p} \Cov{ \bX^*_t B_{j \cdot}, \bX_t B_{i \cdot} } \Cov{ \bX_t B_{i \cdot}, \bX^*_t B_{j \cdot} } \right \}  \bA'_2, \label{eqn:M2}
\end{eqnarray}

$\bM_{A_1}$ and $\bM_{A_2}$ share the same $d$ positive eigenvalues and $\bM_{A_l} \bq = \bzero$ for any vector $\bq$ perpendicular to $\calM(\bA_l)$, $l=1,2$. Therefore, the columns of $\calM(\bA_l)$, $l=1,2$, can be estimated as the $d$ orthonormal eigenvectors of matrix $\bM_{A_l}$ correspond to $d$ positive eigenvalues and the columns are arranged such that the corresponding eigenvalues are in the descending order. 

Now we define the sample version of these quantities and introduce the estimation procedure. Suppose we have centered our observations $\bY_{1t}$ and $\bY_{2t}$, let $\hat{\bOmega}_{A,ij}$ be the sample cross-space covariance of $i$-th and $j$-th variables and $\hat{\bM}_{A_l}$ be the sample version of $\bM_{A_l}$, $l=1,2$, that is
\begin{equation} \label{eqn:M_hat}
\hat{\bOmega}_{A,ij} = \frac{1}{T} \sum_{t=1}^{T} Y_{1t, \cdot i} Y'_{2t, \cdot j}, \quad \hat{\bM}_{A_1} = \sum_{i=1}^{p} \sum_{j=1}^{p} \hat{\bOmega}_{A,ij} \hat{\bOmega}'_{A,ij}, \quad \hat{\bM}_{A_2} = \sum_{i=1}^{p} \sum_{j=1}^{p} \hat{\bOmega}'_{A,ij} \hat{\bOmega}_{A,ij}.
\end{equation}
A natural estimator for $\bA_l$ is defined as $\hat{\bA}_l = \{ \hat{\ba}_{l1}, \cdots, \hat{\ba}_{ld} \}$, $l=1,2$, where $\hat{\ba}_{lj}$ is the eigenvector of $\hat{\bM}_{A_l}$ corresponding to its $j$-th largest eigenvalue. However such an estimator ignores the fact that $\bxi_t(\bs)$ is continuous over the set $\calS$. 

\subsection{Estimation for the Variable Loading Matrix $\bB$} \label{subsec:est:B}

To estimate the $p \times r$ variable loading matrix $\bB$, we follow closely the method proposed by \cite{wang2017factor} and work with discrete observations of (\ref{eqn:stvp_c0}) at $n$ sampling sites. Let the vector observed at site $\bs_i$ at time $t$ be $\by_t(\bs_i)$. The temporal-cross-covariance between observations from site $\bs_i$ and $\bs_j$ for lag $h \ge 1$ is
\begin{equation}  \label{eqn:omega_B}
\bOmega_{B,ij}(h) = \Cov{\by_t(\bs_i), \; \by_{t+h}(\bs_j)} = \bB \, \Cov{\bX'_t \ba(\bs_i), \; \ba'(\bs_j) \bX_t} \, \bB'.
\end{equation}
The last equation results form the assumption that $\bX_t$ is uncorrelated with $\bE_t$ at all leads and lags and $\bE_t$ is white noise. For a pre-determined maximum lag $h_0$, define
\begin{equation}  \label{eqn:M_B}
\bM_B = \sum_{h=1}^{h_0} \sum_{i=1}^{n} \sum_{j=1}^{n} \bOmega_{B,ij}(h) \bOmega'_{B,ij}(h).
\end{equation}
By (\ref{eqn:omega_B}) and (\ref{eqn:M_B}), it follows that
\begin{equation}  \label{eqn:M_B_sandwich}
\bM_B = \bB \, \left( \sum_{h=1}^{h_0} \sum_{i=1}^{n} \sum_{j=1}^{n}  \Cov{\bX'_t \ba(\bs_i), \; \ba'(\bs_j) \bX_t} \Cov{\bX'_t \ba(\bs_j), \; \ba'(\bs_i) \bX_t} \right) \, \bB'.
\end{equation}

$\bM_B$ shares the same $r$ positive eigenvalues and $\bM_B \bq = \bzero$ for any vector $\bq$ perpendicular to $\calM(\bB)$. Therefore, the columns of $\calM(\bB)$ can be estimated as the $r$ orthonormal eigenvectors of matrix $\bM_B$ correspond to $r$ positive eigenvalues and the columns are arranged such that the corresponding eigenvalues are in the descending order. 

Define the sample version of $\bOmega_{B,ij}(h)$ and $\bM_B$ for centered observation $\bY_{t}$ as 
\begin{equation} \label{eqn:M_B_hat}
\hat{\bOmega}_{B,ij} = \frac{1}{T-h} \sum_{t=1}^{T-h} Y_{1t, \cdot i} Y'_{2\,t+h, \cdot j}, \quad \hat{\bM}_B = \sum_{h=1}^{h_0} \sum_{i=1}^{n} \sum_{j=1}^{n} \hat{\bOmega}_{B,ij} \hat{\bOmega}'_{B,ij}.
\end{equation}

A natural estimator for $\bB$ can be obtained as $\hat{\bB} = \{ \hat{\bb}_1, \cdots, \hat{\bb}_r \}$, where $\hat{\bb}_i$ is the eigenvector of $\hat{\bM}_B$ corresponding to its $i$-th largest eigenvalue. 

\subsection{Estimation for the Latent Factor Matrix $\bX_t$ and Signal Matrix $\bXi_t$} \label{subsec:est:latentfac&signal}

By (\ref{eqn:stvp_mat_div}), the estimators of two representations of the latent matrix factor $\bX_t$ are defined as
\begin{equation} \label{eqn:fac_mat_est_diff}
\hat{\bX}_t = \hat{\bA}'_1 \bY_{1t} \hat{\bB}, \qquad \hat{\bX}^*_t = \hat{\bA}'_2 \bY_{2t} \hat{\bB}. 
\end{equation}
The latent signal process are estimated by 
\begin{equation} \label{eqn:signal_mat_est_sep}
\hat{\bXi}_t = \begin{bmatrix}\hat{\bXi}_{1 t} \\ \hat{\bXi}_{2 t} \end{bmatrix},
\end{equation}
where
\begin{equation*} 
\hat{\bXi}_{1 t} = \hat{\bA}_1 \hat{\bX}_t \hat{\bB}' = \hat{\bA}_1 \hat{\bA}'_1 \bY_{1t} \hat{\bB} \hat{\bB}', \qquad  \hat{\bXi}_{2 t} = \hat{\bA}_2 \hat{\bX}^*_t \hat{\bB}' = \hat{\bA}_2 \hat{\bA}'_2 \bY_{2t} \hat{\bB} \hat{\bB}'.
\end{equation*}

\subsection{Estimation of the Spatial Loading Matrix $\bA$ and Loading Function $\bA(\bs)$}

Note that now we only have estimated spatial loading matrices $\hat{\bA}_1$ and $\hat{\bA}_2$ on two partitioned set of sampling locations under the constraint that $\bA_1'\bA_1=\bA_2'\bA_2=\bI_d$. Estimate loading functions from $\hat{\bA}_1$ and $\hat{\bA}_2$ separately will result in inefficient use of sampling locations. Also, the constraint that $\bA_1'\bA_1=\bA_2'\bA_2=\bI_d$ complicates the estimation of the loading functions $\ba_j(\bs)$. In addition, (\ref{eqn:fac_mat_est_diff}) gives estimators for two different representations of the latent matrix factor $\bX_t$. To get estimators of spatial loading matrix $\bA$ for all sampling locations and $\bX_t$, we use the estimated $\hat{\bXi}_t$ to re-estimate $\hat{\bA}$ and $\hat{\bX}_t$.  

The population signals process is $\bxi_t(\bs) = \bB \sum_{j=1}^{d} a_j(\bs) \bx_{tj} = \bB \bX'_t \ba(\bs)$. The $n \times p$ matrix $\bXi_t = \bA \bX_t \bB$ is the signal matrix at discretized sampling locations at each time $t$. To reduce dimension, we consider the $n \times r$ variable-factor matrix $\bPsi_t = \bXi_t \bB' = \bA \bX_t$. Let $\bX = \begin{pmatrix} \bX_1 & \cdots & \bX_T \end{pmatrix}$ and $\bPsi = \begin{pmatrix} \bPsi_1 & \cdots & \bPsi_T \end{pmatrix} = \bA \bX$, then 
\[
\frac{1}{nprT} \bPsi' \bPsi = \frac{1}{nprT} \bX' \bA' \bA \bX.
\]

Let the rows of $\frac{1}{\sqrt{rT}}\bW$ be the eigenvectors of $\frac{1}{nprT}\bPsi' \bPsi$ corresponding to its $d$ non-zero egienvalues. The column space of $\bX'$ can be estimated as that of $\bW'$. And $\bA^* = \frac{1}{rT} \bPsi \bW'$ is the loading function values at discretized sampling site corresponding to $\bW$.

However, true $\bXi_t$'s or $\bPsi_t$'s are not observable and only the estimated values $\hat{\bXi}_t$ and $\hat{\Psi} = \hat{\bXi}_t \hat{\bB} $ are available. Thus, we estimate $\frac{1}{\sqrt{rT}} \hat{\bW}$ whose columns are the eigenvectors of $\frac{1}{nprT} \hat{\bPsi}' \hat{\bPsi}$ corresponding to its $d$ non-zero egienvalues and $\hat{\bA} = \frac{1}{rT} \hat{\bPsi} \hat{\bW}'$. The reason that $\hat{\bPsi}$ is choosen over $\hat{\bXi}$ is that $\hat{\bPsi}$ has the same esimaton error bound but is of lower dimension. 

Once $\hat{\bA}$ is estimated, we estimate loading functions $a_j(\bs)$ from the estimated $n$ observations in column $\hat{A}_{\cdot j}$ by the sieve approximation. Any set of bivariate basis functions can be chosen. In our procedure, we consider the tensor product linear sieve space $\Theta_n$, which is constructed as a tensor product space of some commonly used univariate linear approximating spaces, such as B-spline, orthogonal wavelets and polynomial series. Then for each $j \le d$,
\[
a_j(\bs) = \sum_{i=1}^{J_n} \beta_{i,j} u_i(\bs) + r_j(\bs).
\]
Here $\beta_{i,j}$'s are the sieve coefficients of $i$ basis function $u_i(\bs)$ corresponding to the $j$-th factor loading function;  $r_j(\bs)$ is the sieve approximation error; $J_n$ represents the number of sieve terms which grows slowly as $n$ goes to infinity. We estimate $\hat{\beta}_{i,j}$'s and the loading functions are approximated by $\hat{a}_j(\bs) = \sum_{i=1}^{J_n} \hat{\beta}_{i,j} u_i(\bs)$.

\section{Prediction} \label{sec:prediction}

\subsection{Spatial Prediction}

A major focus of spatio-temporal data analysis is the prediction of variable of interest over new locations. For some new location $\bs_0 \in \calS$ and $\bs_0 \ne \bs_i$ for $i=1,\ldots,n$, we aim to predict the unobserved value $\by_t(\bs_0)$, $t=1,\ldots,T$, based on observations $\bY_t$. By (\ref{eqn:stvp_c0}), we have $\by_t(\bs_0) = \bxi_t(\bs_0) + \bepsilon_t(\bs_0) = \bB \bX'_t \ba(\bs_0) + \bepsilon_t(\bs_0)$. As recommended by \cite{cressie2015statisticsST}, we predict $\bxi_t(\bs_0)=\bB \bX'_t \ba(\bs_0)$ instead of $\by_t(\bs_0)$ directly. Thus, a natural estimator is 
\begin{equation} \label{eqn:pred_xi_s0}
	\hat{\bxi}_t(\bs_0)=\hat{\bB} \hat{\bX}'_t \hat{\ba}(\bs_0),
\end{equation}
where $\hat{\bB}$, $\hat{\bX}$ and $\hat{\ba}(\bs)$ are estimated following procedures in Section \ref{sec:est}.

\subsection{Temporal Prediction}

Temporal prediction focuses on predict the future values $\by_{t+h}(\bs_1), \ldots, \by_{t+h}(\bs_n)$ for some $h \ge 1$. By (\ref{eqn:stvp_c0}), we have $\by_{t+h}(\bs) = \bxi_{t+h}(\bs) + \bepsilon_{t+h}(\bs) = \bB \bX'_{t+h} \ba(\bs) + \bepsilon_{t+h}(\bs)$. Since $\bepsilon_{t+h}(\bs)$ is unpredictable white noise, the ideal predictor for $\by_{t+h}(\bs)$ is that for $\bxi_{t+h}(\bs)$. Thus, we focus on predict $\bxi_{t+h}(\bs) = \bB \bX'_{t+h} \ba(\bs)$. The temporal dynamics of the $\bxi_{t+h}(\bs)$ present in a lower dimensional matrix factor $\bX'_{t+h}$, thus a more effective approach is to predict $\bX'_{t+h}$ based on $\bX'_{t-l}, \ldots, \bX'_t$ where $l$ is a prescribed integer. The rows and columns of $\bX_t$ represents the spatial factors and the variable factor, respectively. To preserve the matrix structure intrinsic to $\bX_t$, we model $\{\bX_t\}_{1:T}$ as the matrix autoregressive model of order one. Mathematically,
\begin{equation}  \label{eqn:mar1}
\bX_{t} = \bPhi_R \, \bX_{t-1} \, \bPhi_C + \bU_t,
\end{equation}
where $\Phi_R$ and $\Phi_C$ are row and column coefficient matrices, respectively. The covariance structure of the matrix white noise $\bU_t$ is not restricted. Thus, $vec{\bU_t} \sim \calN(\bzero, \bSigma_U)$ where $\bSigma_U$ is an arbitrary covariance matrix. Matrix $\Phi_R$ captures the auto-correlations between the spatial latent factors and $\Phi_C$ captures the auto-correlations between the variable latent factors. 

Following the generalized iterative method proposed in \cite{yang2017autoregressive}, we have estimators $\hat{\bPhi}_R$ and $\hat{\bPhi}_C$. The prediction for $\by_{t+h}(\bs)$ is best approximate by \begin{equation}
\hat{\bxi}_{t+h}(\bs)=\hat{\bB} \, \hat{\bX}'_{t+h} \, \hat{\ba}(\bs)=\hat{\bB} \, \hat{\bPhi}^h_R \, \hat{\bX}_{t} \, \hat{\bPhi}^h_C \, \hat{\ba}(\bs),
\end{equation}
where $\hat{\bB}$, $\hat{\bX}$ and $\hat{\ba}(\bs)$ are estimated following procedures in Section \ref{sec:est} and $\hat{\bPhi}^h_R$ and $\hat{\ba}(\bs)$ is estimated from MAR(1) model. 

\section{Asymptotic properties}  \label{sec:theory}

In this section, we investigate the rates of convergence for the estimators under the setting that $n$, $p$ and $T$ all go to infinity while $d$ and $r$ are fixed and the factor structure does not change over time. In what follows, let $\norm{\bA}_2 = \sqrt{\lambda_{max}(\bA'\bA)}$ and $\norm{\bA}_F = \sqrt{tr(\bA'\bA)}$ denote the spectral and Frobenius norms of the matrix $\bA$, respectively. $\norm{\bA}_{min}$ denotes the positive square root of the minimal eigenvalue of $\bA'\bA$ or $\bA\bA'$, whichever is a smaller matrix. When $\bA$ is a square matrix, we denote by $tr(\bA)$, $\lambda_{max}(\bA)$ and $\lambda_{max}(\bA)$ the trace, maximum and minimum eigenvalues of the matrix $\bA$, respectively. For two sequences $a_N$ and $b_N$, we write $a_N \asymp b_N$ if $a_N = O(b_N)$ and $b_N = O(a_N)$. The following regularity conditions are imposed before we derive the asymptotics of the estimators.  

\begin{condition}  \label{cond:vecXt_alpha_mixing}  
\textbf{Alpha-mixing.} $\{ \vect{\bX_t}, t=0,\pm 1, \pm 2, \cdots \}$ is strictly stationary and $\alpha$-mixing. Specifically, for some $\gamma > 2$, the mixing coefficients satisfy the condition that $\sum_{h=1}^{\infty} \alpha(h)^{1-2/\gamma} < \infty$, where $ \alpha(h)=\underset{\tau}{\sup} \underset{A \in \mathcal{F}_{-\infty}^{\tau}, B \in \mathcal{F}_{\tau+h}^{\infty}}{\sup} \left| P(A \cap B) - P(A)P(B) \right| $ and $\mathcal{F}_{\tau}^s$ is the $\sigma$-field generated by $\{vec(\bX_t): {\tau} \le t \le s \}$. 
\end{condition}

\begin{condition}  \label{cond:Xt_cov_fullrank_bounded}
Let $X_{t, ij}$ be the $ij$-th entry of $\bX_t$. Then, $E(\left| X_{t, ij} \right|^{2 \gamma}) \le C$ for any $i = 1, \ldots, d$, $j = 1, \ldots, r$ and $t = 1, \ldots, T$, where $C$ is a positive constant and $\gamma$ is given in Condition \ref{cond:vecXt_alpha_mixing}. In addition, there exists an integer $h$ satisfying $1 \le h \le h_0$ such that $\bSigma_f(h)$ is of rank $k = \max(d, r)$ and $\norm{\bSigma_f(h)}_2 \asymp O(1) \asymp \sigma_k(\bSigma_f(h))$. For $i = 1, \ldots, d$ and $j=1, \ldots, r$, $\frac{1}{T-h} \sum_{t=1}^{T-h} Cov(X_{t, i \cdot}, X_{t+h, i \cdot}) \ne \mathbf{0}$ and $\frac{1}{T-h} \sum_{t=1}^{T-h} Cov(X_{t, \cdot j}, X_{t+h, \cdot j}) \ne \mathbf{0}$.
\end{condition}

\begin{condition}  \label{cond:A_factor_strength}
\textbf{Spacial factor strength.} For any partition $\{ \calS_1, \calS_2 \}$ of locations $\calS = \{ \bs_1, \ldots, \bs_n\}$, there exists a constant $\delta \in [0, 1]$ such that  $\norm{\bA_1}^2_{min} \asymp n_1^{1-\delta} \asymp \norm{\bA_1}^2_2$ and $\norm{\bA_2}^2_{min} \asymp n_2^{1-\delta} \asymp \norm{\bA_2}^2_2$, where $n_1$ and $n_2$ are number of locations in sets $\calS_1$ and $\calS_2$, respectively, and $n_1 + n_2 = n$.
\end{condition}

\begin{condition}  \label{cond:B_factor_strength}
\textbf{Variable factor strength.} There exists a constant $\gamma \in [0,1]$ such that $\norm{\bB}^2_{min} \asymp p^{1-\gamma} \asymp \norm{\bB}^2_2$ as $p$ goes to infinity and $r$ is fixed.
\end{condition}

\begin{condition}
\textbf{Loading functions belongs to H{\"o}lder class.} For $j=1, \ldots, d$, the loading functions $\ba_j(\bs)$, $\bs \in \calS \in \mathbb{R}^2$ belongs to a H{\"o}lder class $\calA^{\kappa}_c(\calS)$ ($\kappa$-smooth) defined by  
\[
\calA^{\kappa}_c(\calS) = \left \{a \in \calC^m(\calS): \underset{[\eta]\le m}{\sup} \; \underset{\bs \in \calS}{\sup} \left| D^{\eta}\, a(\bs) \right| \le c, \text{ and }
\underset{[\eta]=m}{\sup} \; \underset{\bu, \bv \in \calS}{\sup} \frac{\left| D^{\eta}\, a(\bu) - D^{\eta}\, a(\bv) \right|}{\norm{\bu - \bv}^{\alpha}_2}  \le c   \right \}, 
\]
for some positive number $c$. Here, $\calC^m(\calS)$ is the space of all $m$-times continuously differentiable real-value functions on $\calS$. The differential operator $D^{\eta}$ is defined as $D^{\eta} = \frac{\partial^{[\eta]}}{\partial s_1^{\eta_1} \partial s_2^{\eta_2}}$ and $[\eta] = \eta_1 + \eta_2$ for nonnegative integers $\eta_1$ and $\eta_2$. 

\end{condition}

Theorem \ref{thm:A_sep_err_bnd} presents the error bound for estimated loading matrix $\bA_1$ and $\bA_2$. 

\begin{theorem} \label{thm:A_sep_err_bnd}
Under Condition \ref{cond:vecXt_alpha_mixing}-\ref{cond:B_factor_strength} and $n^{\delta} p^{\gamma} T^{-1/2} = o(1)$, we have 
\begin{equation}
\calD \left( \calM(\hat{\bA}_i), \calM(\bA_i) \right) = O_p( ( n_1 n_2^{\delta-1} p^{\gamma} + n_1^{\delta-1} n_2 p^{\gamma} + n_1^{\delta} n_2^{\delta} p^{2\gamma} ) T^{-1} )^{1/2}.
\end{equation}
If $n_1 \asymp n_2 \asymp n$, we have
\begin{equation}
\calD \left( \calM(\hat{\bA}_i), \calM(\bA_i) \right) = O_p(  n^{\delta} p^{\gamma} T^{-1/2} ).
\end{equation}
\end{theorem}

Theorem \ref{thm:signal_err_bnd} presents the error bound for estimated signal $\hat{\bXi}_{it}$ and $\hat{\bXi}_t$. 

\begin{theorem} \label{thm:signal_err_bnd}
This proposition considers the error bound of signal estimator as in (\ref{eqn:signal_est}) for each partition. Under $n^{\delta} p^{\gamma} T^{-1}=o_p(1)$, if $n_1 \asymp n2 \asymp n$, then
\begin{equation}
n^{-1/2} p^{-1/2} \norm{\hat{\bXi}_{it} - \bXi_{it}}_2  = O_p(n^{\delta/2} p^{\gamma/2} T^{-1/2} + n^{-1/2} p^{-1/2}),
\end{equation}
for $i = 1, 2$, and 
\begin{equation}
n^{-1} p^{-1} \norm{\hat{\bXi}_{t} - \bXi_{t}}^2_2  = O_p(n^{\delta} p^{\gamma} T^{-1} + n^{-1/2+\delta/2} p^{-1/2+\gamma/2} T^{-1/2} + n^{-1} p^{-1})
\end{equation}
\end{theorem}

Let $\Delta_{npT} = n^{\delta} p^{\gamma} T^{-1} + n^{-1/2+\delta/2} p^{-1/2+\gamma/2} T^{-1/2} + n^{-1} p^{-1}$. Theorem \ref{thm:reest_X_err_bnd} presents the error bound for re-estimated latent factor $\frac{1}{rT}\bW_t$ whose columns are assume to be the eigenvectors of $\frac{1}{rT} \bPsi' \bPsi$. And Proposition \ref{thm:reest_A_err_bnd} presents the error bound for re-estimated whole loading matrix $\bA$ corresponding to estimated $\bW$.

\begin{theorem} \label{thm:reest_X_err_bnd}
\[
\frac{1}{rT} \norm{\hat{\bW}' - \bW'}^2_F = O_p \left(\Delta_{npT} + n^{\delta} p^{\gamma}\Delta_{npT}^2\right)
\]
\end{theorem}

Proposition \ref{thm:reest_A_err_bnd} presents the error bond for estimated spatial loading matrix $\hat{\bA}$.

\begin{proposition} \label{thm:reest_A_err_bnd}
\[
\frac{1}{np} \norm[\big]{\hat{\bA} - \bA}^2_F = O_p \left( \Delta_{npT} \right). 
\]
\end{proposition}

Theorem \ref{thm:space_krig_mse_err_bnd} presents the space kriging error bound based on sieve approximated function $\hat{\bA}(\bs)$.

\begin{theorem} \label{thm:space_krig_mse_err_bnd}
\begin{equation}
\frac{1}{pT} \norm{\hat{\bxi}(\bs_0) - \bxi(\bs_0)}^2_2 = O_p(J_n^{-2\kappa} n^{-\delta} p^{-\gamma} + \Delta_{npT} + 1/T)
\end{equation}
\end{theorem}

\section{Simulation} \label{sec:simu}
In this section we study the numerical performance of the proposed method on synthetic datasets. We let $\bs_1, \cdots, \bs_n$ be drawn randomly from the uniform distribution on $[-1,1]^2$ and the observed data $\by_t(\bs)$ be generated according to model (\ref{eqn:stvp_c0}),
\begin{equation} 
\by_t(\bs) = \bxi_t(\bs) + \bepsilon_t(\bs) = \bB \bX'_t \ba(\bs) + \bepsilon_t(\bs). \nonumber  
\end{equation}
The dimensions of $\bX_t$ are chosen to be $d=3$, $r=2$, and are fixed in all simulations. The latent factor $\bX_t$ is generated from the Gaussian matrix time series (\ref{eqn:mar1})
\begin{equation} 
\bX_{t} = \bPhi_R \, \bX_{t-1} \, \bPhi_C + \bU_t, \nonumber
\end{equation}
where $\bPhi_R = diag(0.7, \, 0.8, \, 0.9)$, $\bPhi_C = diag(0.8, \, 0.6)$ and the entries of $\bU_t$ are white noise Gaussian process with mean $\bzero$ and covariance structure such that $\bSigma_U = \Cov{vec(\bU_t)}$:
\begin{itemize}
\item Model I: $\bSigma_U=\bI_{dr}$. (now)
\item Model II: Kronecker product covariance structure $\bSigma_U=\bSigma_C \otimes \bSigma_R$, where $\bSigma_R$ and $\bSigma_C$ are of sizes $d \times d$ and $r \times r$, respectively. Both $\bSigma_R$ and $\bSigma_C$ have values 1 on the diagonal entries and 0.2 on the off-diagonal entries.
\item Model III: Arbitrary covariance matrix $\bSigma_U$.
\end{itemize}

The entries of $\bB$ is independently sampled from the uniform distribution $\calU(-1,1) \cdot p^{\gamma/2}$. The nugget process $\bepsilon_t(\bs)$ are independent and normal with mean $\bzero$ and the covariance $(1+s_1^2+s_2^2)/2\sqrt{3}\cdot \bI_{p}$. The basis functions $a_j(\bs)$'s are designed to be 
\[
a_1(\bs) = (s_1 - s_2)/2, \quad a_2(\bs)=\cos \left( \pi \sqrt{2(s_1^2+s_2^2)} \right), \quad a_3(\bs)=1.5s_1s_2. 
\]

With the above generating model setting, the signal-noise-ratio of $p$-dimensional variable, which is defined as
\[ 
SNR \equiv \frac{ \int_{\bs \in [-1,1]^2} Trace \left[ Cov  \left(\bxi_{t}(\bs) \right) \right] d\bs  }{ \int_{\bs \in [-1,1]^2} Trace \left[ Cov \left(\bepsilon_{t}(\bs)\right) \right] d\bs } \approx 2.58.
\]

We run $200$ simulations for each combination of $n = 50, 100, 200, 400$, $p = 10, 20, 40$, and $T= 60, 120, 240$. With each simulation, we calculate $\hat{d}$, $\hat{r}$, $\hat{\bA}_1$, $\hat{\bA_2}$, $\hat{\bB}$ and $\hat{\bXi}_t$, reestimate $\hat{\bA}$ and $\widetilde{\bXi}_t$, then use $\hat{\bA}$ to get approximated $\hat{a}_j(\bs)$ following the estimation procedure described in Section \ref{sec:est}.

Table \ref{table:freq_dr_est} presents the relative frequencies of estimated rank pairs over 200 simulations. The columns corresponding to the true rank pair $(3,2)$ is highlighted.  

The performance of correctly estimating the loading spaces are measured by the space distance between the estimated and true loading matrices $\hat{\bA}$ and $\bA$, which is defined as 
\[
\calD(\calM(\hat{\bA}), \calM(\bA)) = \left( 1 - \frac{1}{\max(d,\hat{d})} tr\left( \hat{\bA} (\hat{\bA}'\hat{\bA})^{-1} \hat{\bA}' \cdot \bA (\bA'\bA)^{-1} \bA' \right) \right)^{\frac{1}{2}}.
\]
It can be shown that $\calD(\calM(\hat{\bA}), \calM(\bA))$ takes its value in $[0,1]$, it equals to $0$ if and only if $\calM(\hat{\bA}) = \calM(\bA)$, and equals to $1$ if and only if $\calM(\hat{\bA}) \perp \calM(\bA)$.

\begin{table}[htbp!]
\centering
\caption{Relative frequency of estimated rank pair $(\hat{d}, \hat{r})$ over 200 simulations. The columns correspond to the true value pair $(3,2)$ are highlighted. Blank cell represents zero value.}
\label{table:freq_dr_est}
\resizebox{0.9\textwidth}{!}{%
\begin{tabular}{ccc|ccccc|cccccc}
\hline
\multicolumn{3}{c|}{$(\hat{d}, \, \hat{r})$} & \multicolumn{5}{c|}{$\gamma = 0$} & \multicolumn{6}{c}{$\gamma = 0.5$} \\ \hline
T & p & n & \cellcolor[HTML]{EFEFEF}(3,2) & (3,1) & (2,2) & (1,2) & (1,1) & \cellcolor[HTML]{EFEFEF}(3,2) & (3,1) & (2,2) & (2,1) & (1,2) & (1,1) \\ \hline
60 & 10 & 50 & \cellcolor[HTML]{EFEFEF}0.74 & 0.04 & 0.04 & 0.18 & 0.02 & \cellcolor[HTML]{EFEFEF}0.11 & 0.01 & 0.13 & 0.01 & 0.61 & 0.14 \\
120 & 10 & 50 & \cellcolor[HTML]{EFEFEF}0.93 & 0.07 &  &  & 0.01 & \cellcolor[HTML]{EFEFEF}0.37 & 0.05 & 0.06 & 0.02 & 0.42 & 0.09 \\
240 & 10 & 50 & \cellcolor[HTML]{EFEFEF}0.95 & 0.06 &  &  &  & \cellcolor[HTML]{EFEFEF}0.82 & 0.10 & 0.01 &  & 0.07 & 0.02 \\ \hdashline
60 & 20 & 50 & \cellcolor[HTML]{EFEFEF}0.86 &  & 0.02 & 0.13 &  & \cellcolor[HTML]{EFEFEF}0.02 &  & 0.10 &  & 0.88 & 0.01 \\
120 & 20 & 50 & \cellcolor[HTML]{EFEFEF}1.00 &  &  &  &  & \cellcolor[HTML]{EFEFEF}0.08 &  & 0.04 &  & 0.88 &  \\
240 & 20 & 50 & \cellcolor[HTML]{EFEFEF}1.00 &  &  &  &  & \cellcolor[HTML]{EFEFEF}0.49 &  & 0.01 &  & 0.50 &  \\ \hdashline
60 & 40 & 50 & \cellcolor[HTML]{EFEFEF}0.96 &  & 0.01 & 0.04 &  & \cellcolor[HTML]{EFEFEF}0.03 &  & 0.09 &  & 0.89 &  \\
120 & 40 & 50 & \cellcolor[HTML]{EFEFEF}1.00 &  &  &  &  & \cellcolor[HTML]{EFEFEF}0.02 &  & 0.07 &  & 0.91 &  \\
240 & 40 & 50 & \cellcolor[HTML]{EFEFEF}1.00 &  &  &  &  & \cellcolor[HTML]{EFEFEF}0.32 &  & 0.01 &  & 0.68 &  \\ \hline
60 & 10 & 100 & \cellcolor[HTML]{EFEFEF}0.94 & 0.04 & 0.02 &  &  & \cellcolor[HTML]{EFEFEF}0.64 & 0.11 & 0.20 & 0.02 & 0.03 & 0.01 \\
120 & 10 & 100 & \cellcolor[HTML]{EFEFEF}0.96 & 0.05 &  &  &  & \cellcolor[HTML]{EFEFEF}0.93 & 0.07 &  & 0.01 &  &  \\
240 & 10 & 100 & \cellcolor[HTML]{EFEFEF}0.97 & 0.03 &  &  &  & \cellcolor[HTML]{EFEFEF}0.94 & 0.06 &  &  &  &  \\ \hdashline
60 & 20 & 100 & \cellcolor[HTML]{EFEFEF}1.00 &  &  &  &  & \cellcolor[HTML]{EFEFEF}0.73 &  & 0.22 &  & 0.06 &  \\
120 & 20 & 100 & \cellcolor[HTML]{EFEFEF}1.00 &  &  &  &  & \cellcolor[HTML]{EFEFEF}0.97 &  & 0.04 &  &  &  \\
240 & 20 & 100 & \cellcolor[HTML]{EFEFEF}1.00 &  &  &  &  & \cellcolor[HTML]{EFEFEF}1.00 &  &  &  &  &  \\ \hdashline
60 & 40 & 100 & \cellcolor[HTML]{EFEFEF}1.00 &  &  &  &  & \cellcolor[HTML]{EFEFEF}0.72 &  & 0.24 &  & 0.05 &  \\
120 & 40 & 100 & \cellcolor[HTML]{EFEFEF}1.00 &  &  &  &  & \cellcolor[HTML]{EFEFEF}0.96 &  & 0.04 &  &  &  \\
240 & 40 & 100 & \cellcolor[HTML]{EFEFEF}1.00 &  &  &  &  & \cellcolor[HTML]{EFEFEF}1.00 &  &  &  &  &  \\ \hline  
60 & 10 & 200 & \cellcolor[HTML]{EFEFEF}0.98 & 0.03 &  &  &  & \cellcolor[HTML]{EFEFEF}0.84 & 0.11 & 0.03 &  & 0.03 & 0.01 \\
120 & 10 & 200 & \cellcolor[HTML]{EFEFEF}0.97 & 0.04 &  &  &  & \cellcolor[HTML]{EFEFEF}0.94 & 0.07 &  &  &  &  \\
240 & 10 & 200 & \cellcolor[HTML]{EFEFEF}0.97 & 0.03 &  &  &  & \cellcolor[HTML]{EFEFEF}0.95 & 0.05 &  &  &  &  \\  \hdashline
60 & 20 & 200 & \cellcolor[HTML]{EFEFEF}1.00 &  &  &  &  & \cellcolor[HTML]{EFEFEF}0.94 &  & 0.02 &  & 0.04 &  \\
120 & 20 & 200 & \cellcolor[HTML]{EFEFEF}1.00 &  &  &  &  & \cellcolor[HTML]{EFEFEF}1.00 &  &  &  &  &  \\
240 & 20 & 200 & \cellcolor[HTML]{EFEFEF}1.00 &  &  &  &  & \cellcolor[HTML]{EFEFEF}1.00 &  &  &  &  &  \\  \hdashline
60 & 40 & 200 & \cellcolor[HTML]{EFEFEF}1.00 &  &  &  &  & \cellcolor[HTML]{EFEFEF}0.97 &  & 0.01 &  & 0.03 &  \\
120 & 40 & 200 & \cellcolor[HTML]{EFEFEF}1.00 &  &  &  &  & \cellcolor[HTML]{EFEFEF}1.00 &  &  &  &  &  \\
240 & 40 & 200 & \cellcolor[HTML]{EFEFEF}1.00 &  &  &  &  & \cellcolor[HTML]{EFEFEF}1.00 &  &  &  &  &  \\ \hline
60 & 10 & 400 & \cellcolor[HTML]{EFEFEF}0.98 & 0.02 &  &  &  & \cellcolor[HTML]{EFEFEF}0.90 & 0.09 &  &  & 0.02 & 0.01 \\
120 & 10 & 400 & \cellcolor[HTML]{EFEFEF}0.97 & 0.03 &  &  &  & \cellcolor[HTML]{EFEFEF}0.93 & 0.08 &  &  &  &  \\
240 & 10 & 400 & \cellcolor[HTML]{EFEFEF}0.97 & 0.03 &  &  &  & \cellcolor[HTML]{EFEFEF}0.96 & 0.04 &  &  &  &  \\  \hdashline
60 & 20 & 400 & \cellcolor[HTML]{EFEFEF}1.00 &  &  &  &  & \cellcolor[HTML]{EFEFEF}1.00 &  &  &  & 0.01 &  \\
120 & 20 & 400 & \cellcolor[HTML]{EFEFEF}1.00 &  &  &  &  & \cellcolor[HTML]{EFEFEF}1.00 &  &  &  &  &  \\
240 & 20 & 400 & \cellcolor[HTML]{EFEFEF}1.00 &  &  &  &  & \cellcolor[HTML]{EFEFEF}1.00 &  &  &  &  &  \\  \hdashline
60 & 40 & 400 & \cellcolor[HTML]{EFEFEF}1.00 &  &  &  &  & \cellcolor[HTML]{EFEFEF}1.00 &  &  &  & 0.01 &  \\
120 & 40 & 400 & \cellcolor[HTML]{EFEFEF}1.00 &  &  &  &  & \cellcolor[HTML]{EFEFEF}1.00 &  &  &  &  &  \\
240 & 40 & 400 & \cellcolor[HTML]{EFEFEF}1.00 &  &  &  &  & \cellcolor[HTML]{EFEFEF}1.00 &  &  &  &  &  \\ \hline
\end{tabular}%
}
\end{table}

Figure \ref{fig:spdistA.0} presents the box plot of the average space distance
\[
\frac{1}{2}\left( \calD(\calM(\hat{\bA}_1), \calM(\bA_1)) + \calD(\calM(\hat{\bA}_2), \calM(\bA_2))  \right)
\]
and compare it with the box plot of space distance between re-estimated $\hat{\bA}$ and the truth $\bA$. 

Figure \ref{fig:spdistB.0} presents the box plot of the space distance between $\hat{\bB}$ and the truth $\bB$.

\begin{figure}[ht!]
	\centering
	\includegraphics[width=\linewidth,height=\textheight,keepaspectratio=true]{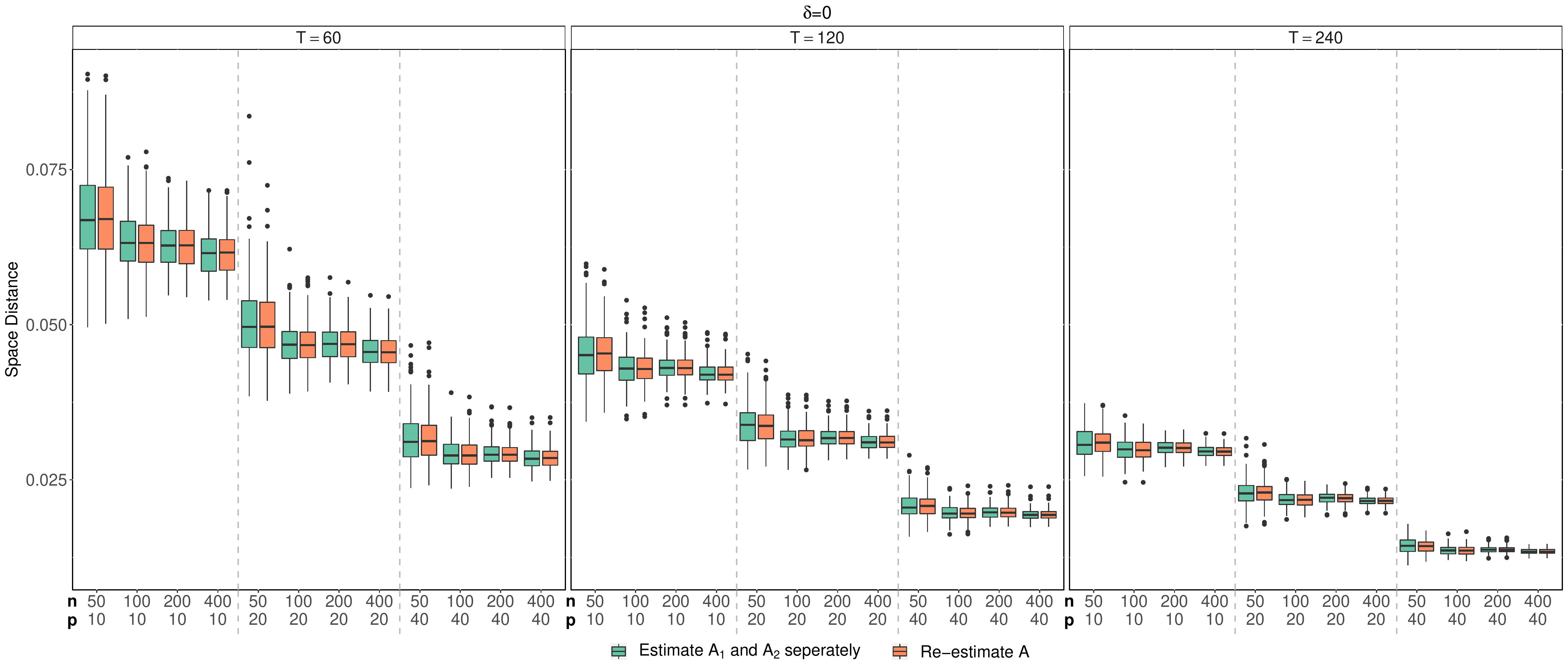}
	\caption{Box-plots of the estimation accuracy measured by $\mathcal{D}(\widehat{\bA}, \bA)$ for the case of orthogonal constraints. Gray boxes represent the average of $\mathcal{D}(\widehat{\bA}_1, \bA_1)$ and $\mathcal{D}(\widehat{\bA}_2, \bA_2)$. The results are based on $200$ iterations. See Table \ref{table:spdist_msd_table} in Appendix \ref{appendix:tableplots} for mean and standard deviations of the spatial distance. }
	\label{fig:spdistA.0}
\end{figure} 

\begin{figure}[ht!]
	\centering
	\includegraphics[width=\linewidth,height=\textheight,keepaspectratio=true]{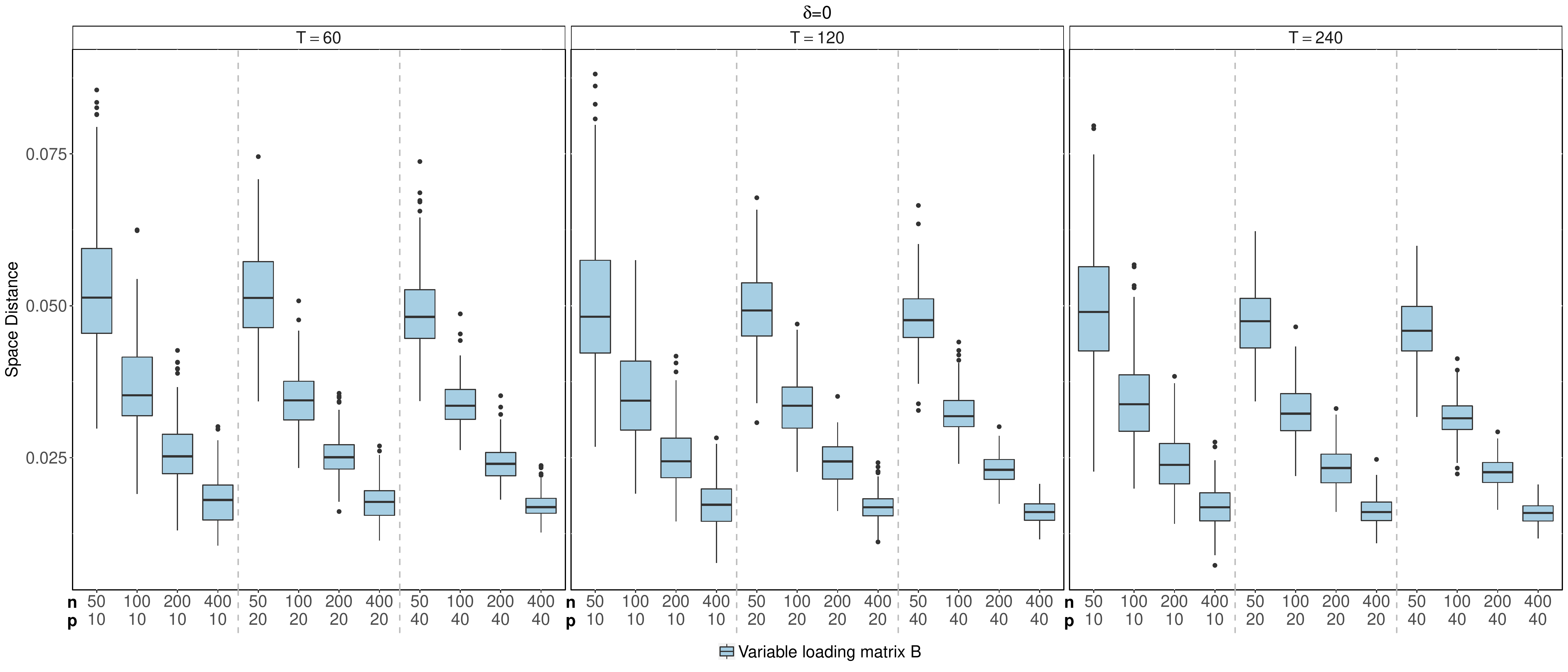}
	\caption{Box-plots of the estimation accuracy of variable loading matrix measured by $\mathcal{D}(\widehat{\bB}, \bB)$. The results are based on $200$ iterations. See Table \ref{table:spdist_msd_table} in Appendix \ref{appendix:tableplots} for mean and standard deviations of the spatial distance. }
	\label{fig:spdistB.0}
\end{figure} 

\begin{figure}[ht!]
	\centering
	\includegraphics[width=\linewidth,height=\textheight,keepaspectratio=true]{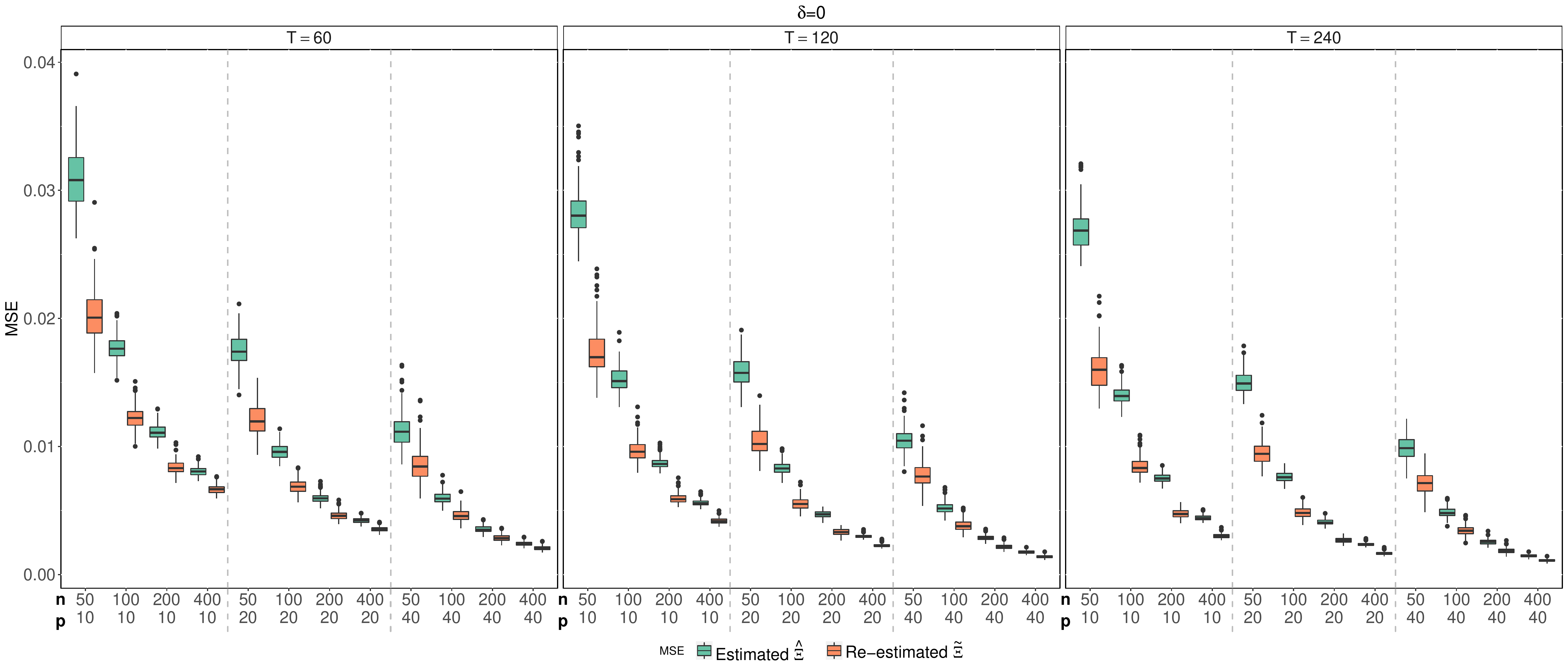}
	\caption{Box-plots of the estimation of signals MSE. Gray boxes represent the our procedure. The results are based on $200$ iterations. See Table \ref{table:spdist_msd_table} in Appendix \ref{appendix:tableplots} for mean and standard deviations of the MSE. }
	\label{fig:signalest.0}
\end{figure} 

\begin{figure}[ht!]
	\centering
	\includegraphics[width=\linewidth,height=\textheight,keepaspectratio=true]{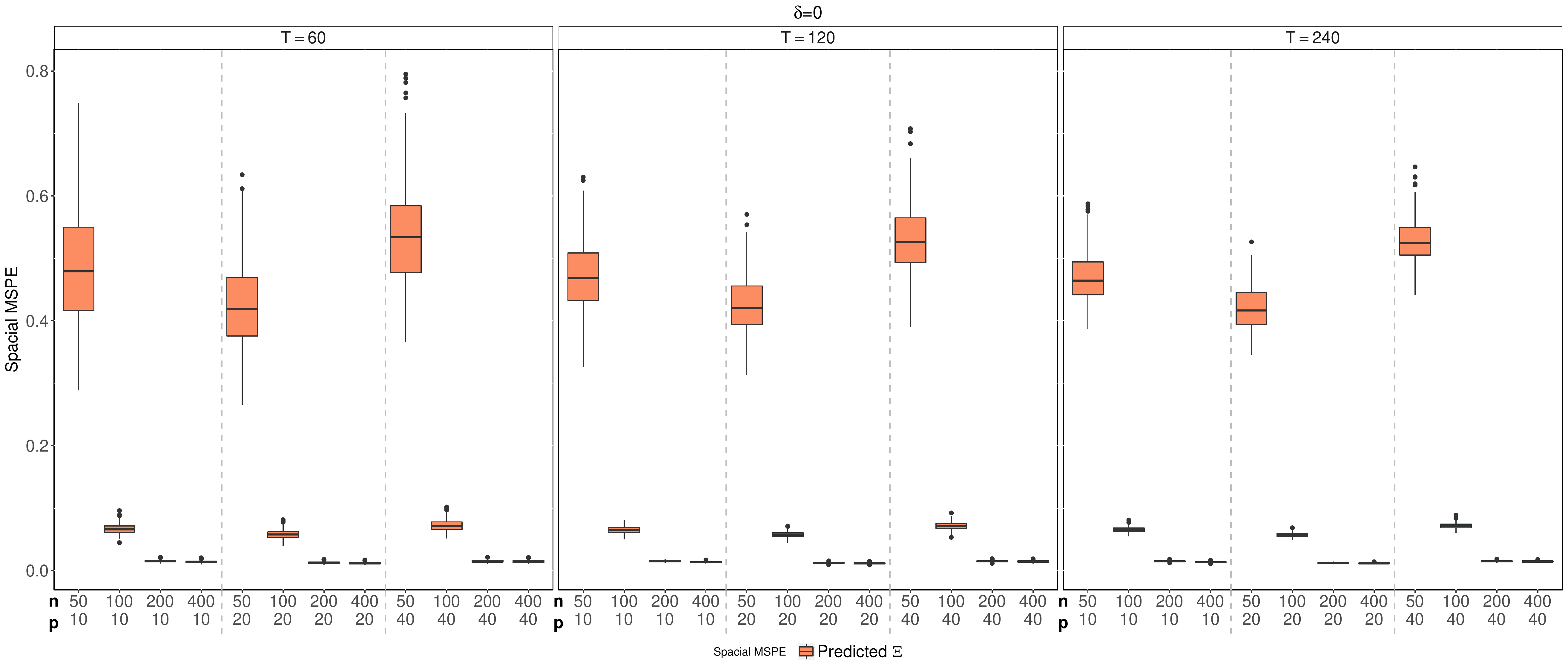}
	\caption{Box-plots of the spatial prediction measured by average MSPE for $50$ new locations. Colored boxes represent the our model. The results are based on $200$ iterations. See Table \ref{table:STprediction} in Appendix \ref{appendix:tableplots} for mean and standard deviations of the MSPE. }
	\label{fig:signalspaceprd.0}
\end{figure} 

\begin{figure}[ht!]
	\centering
	\includegraphics[width=\linewidth,height=\textheight,keepaspectratio=true]{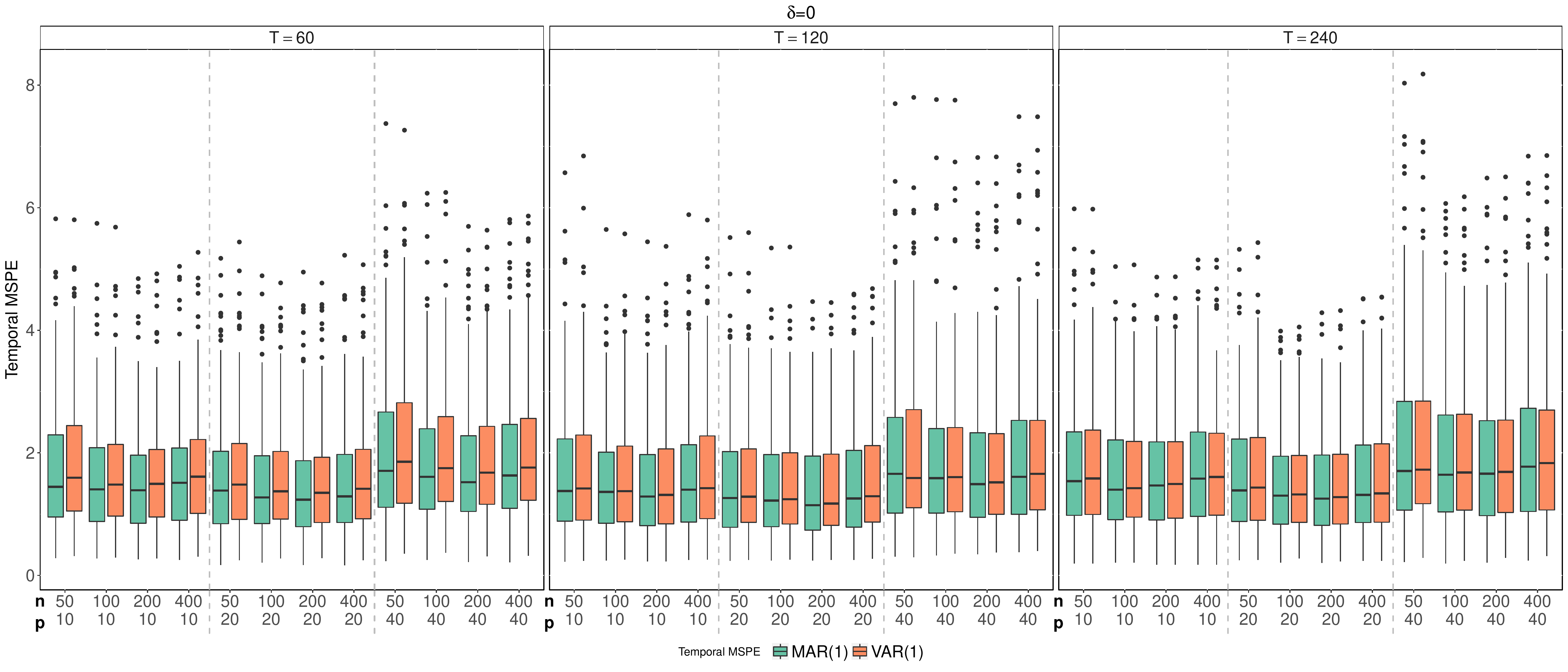}
	\caption{Box-plots of the one step ahead forecasting accuracy measured by MSPE. Gray boxes represent the MAR(1) model. The results are based on $200$ iterations. See Table \ref{table:STprediction} in Appendix \ref{appendix:tableplots} for mean and standard deviations of the MSPE. }
	\label{fig:signaltimeprd.0}
\end{figure}

Define the mean squared error of estimated signals $\hat{\bxi}$ as
\[
MSE(\hat{\bxi}) = \frac{1}{npT} \sum_{t=1}^{T} \sum_{i=1}^{n} \norm{ \hat{\bxi}_t(\bs_i) - \bxi_t(\bs_i) }^2_2.
\]
We compare the mean square error between first estimated $\hat{\bXi}_t$ defined in (\ref{eqn:signal_mat_est_sep}) and re-estimated $\tilde{\bXi}_t$ defined as
\[
\tilde{\bXi} = \begin{bmatrix} \tilde{\bXi}_1, \cdots, \tilde{\bXi}_T \end{bmatrix} = \tilde{\bA} \tilde{\bX} \hat{\bB}'.
\]
The box plots of $MSE(\hat{\bxi})$ and $MSE(\tilde{\bxi})$ are in Figure \ref{fig:signalspaceprd.0}. Re-estimated provides much more accurate estimate for $\bxi_t(\bs_j)$ than $\tilde{\bxi}_t(\bs_j)$ does. 

To demonstrate the performance of spatial prediction, we generate data at a set $\calS_0$ of $50$ new locations randomly sampled from $\calU[-1,1]^2$. For each $t = 1, \ldots, T$, we calculate the spatial prediction $\hat{\by}_t(\cdot) = \hat{\bxi}_t(\cdot)$ defined in (\ref{eqn:pred_xi_s0}) for each location in $\calS_0$. The mean squared spatial prediction error is calculated as 
\[
MSPE(\hat{\by}) = \frac{1}{50 p T} \sum_{t=1}^{T} \sum_{s_0 \in \calS_0} \norm{ \hat{\by}_t(\bs_0) - \bxi_t(\bs_0) }^2_2.
\]  

To demonstrate the performance of temporal forecasting, we generate $\bX_{T+h}$ according to the matrix time series (\ref{eqn:mar1}) for $h=1,2$ and compute both the one-step-ahead and two-step-ahead predictions at time $T$. The mean square temporal prediction error is computed as +
\[
MSPE(\hat{\by}_{T+h}) = \frac{1}{n p} \sum_{j=1}^{n} \norm{ \hat{\by}_{T+h}(\bs_j) - \bxi_{T+h}(\cdot) }^2_2.
\]

Figure \ref{fig:signalspaceprd.0} presents box-plots of the spatial prediction measured by average MSPE for $50$ new locations. The results are based on $200$ iterations. Figure \ref{fig:signaltimeprd.0} compares the MSPEs using matrix time series MAR(1) and vectorized time series VAR(1) estimates.

The means and standard errors of the MSPEs from 200 simulations for each model setting are reported in Table \ref{table:STprediction} in Appendix \ref{appendix:tableplots}. It also reports the means and standard errors of the MSPEs using matrix time series MAR(1) and vectorized time series VAR(1) estimates.

\section{Real Data Application} \label{sec:appl}

In this section, we apply the proposed method to the Comprehensive Climate Dataset (CCDS) -- a collection of climate records of North America. The dataset was compiled from five federal agencies sources by \cite{lozano2009spatial}. It contains monthly observations of 17 climate variables spanning from 1990 to 2001 on a $2.5 \times 2.5$ degree grid for latitudes in $(30.475, 50.475)$, and longitudes in $(-119.75, -79.75)$. The total number of observation locations is 125 and the length of the whole time series is 156. Table \ref{table:CCDS_varlist} lists the variables used in our analysis. Detailed information about data pre-processing is given in \cite{lozano2009spatial}.

\begin{table}[htpb!]
\centering
\caption{Variables and data sources in the Comprehensive Climate Dataset (CCDS)}
\label{table:CCDS_varlist}
\resizebox{0.7\textwidth}{!}{%
\begin{tabular}{l|c|l|c}
\hline
Variables (Short name) & Variable group & \multicolumn{1}{c|}{Type} & Source \\ \hline
Methane (CH4) & $CH_4$ & \multirow{4}{*}{Greenhouse Gases} & \multirow{4}{*}{NOAA} \\
Carbon-Dioxide (CO2) & $CO_2$ &  &  \\
Hydrogen (H2) & $H_2$ &  &  \\
Carbon-Monoxide (CO) & $CO$ &  &  \\ \hline
Temperature (TMP) & TMP & \multirow{8}{*}{Climate} & \multirow{8}{*}{CRU} \\
Temp Min (TMN) & TMP &  &  \\
Temp Max (TMX) & TMP &  &  \\
Precipitation (PRE) & PRE &  &  \\
Vapor (VAP) & VAP &  &  \\
Cloud Cover (CLD) & CLD &  &  \\
Wet Days (WET) & WET &  &  \\
Frost Days (FRS) & FRS &  &  \\ \hline
Global Horizontal (GLO) & SOL & \multirow{4}{*}{Solar Radiation} & \multirow{4}{*}{NCDC} \\
Direct Normal (DIR) & SOL &  &  \\
Global Extraterrestrial (ETR) & SOL &  &  \\
Direct Extraterrestrial (ETRN) & SOL &  &  \\ \hline
Utra Violet (UV) & AER & Aerosol Index & NASA \\ \hline
\end{tabular}%
}
\end{table}

We first remove the trend and annually seasonal component by taking difference between observations from the same month in consecutive years. Then we normalized this data set by removing the trend and dividing it by the standards deviation for each variable across space. We randomly select 10\% of locations and predict the value of all variables over the whole time span for these locations. We repeat the procedure 100 times and report the average spatial MSPE. 

\begin{table}[htpb!]
\centering
\caption{Multivariate kriging comparison.}
\label{table:compare_cokrig}
\begin{tabular}{c|ccc}
\hline
MSPE & LLDF & Simple & Ordinary \\ \hline
Spatial & 0.4812 & 0.7634 & 0.7312 \\
 &  &  &  \\ \hline
\end{tabular}
\end{table}

We compare the spatial prediction performance of our proposed method with the classical cokriging approaches including simple kriging and ordinary cokriging with nonbias condition which are applied to each variable separately. The R package {\it gstat} is used for the classical cokriging algorithms. Comparison of the spatial prediction performance between different methods are presented in Table \ref{table:compare_cokrig}.


\clearpage
\bibliographystyle{\mybibsty}
\bibliography{\mybib}

\begin{appendices}
\include{LCKproofs}

\include{LCKmoretableplots}
\end{appendices}

\end{document}

%% file: ycPre.tex
\usepackage{amsmath, amsthm}  
\usepackage{mathtools}   
\usepackage{bm}
\usepackage{amsfonts}

\usepackage{graphicx} 
\usepackage{enumerate}
\usepackage[authoryear]{natbib}  
\usepackage{url} 

\usepackage{enumitem}   
\usepackage[table,xcdraw,dvipsnames]{xcolor}   
\usepackage{pdflscape}
\usepackage{afterpage}
\usepackage{capt-of}
\usepackage{adjustbox}

\usepackage[many]{tcolorbox}
\usepackage{multirow}   
\usepackage{tabularx} 

\usepackage{float}   
\usepackage{nameref}   

\usepackage[font=normalsize]{caption}  
\usepackage{subcaption}  
\usepackage[titletoc,title]{appendix}  
\usepackage{longtable} 

\usepackage[final]{hyperref} 
\hypersetup{
	colorlinks=true,       
	linkcolor=blue,        
	citecolor=blue,        
	filecolor=magenta,     
	urlcolor=blue         
}

\newcommand{\utwi}[1]{\mbox{\boldmath $ #1$}}

\newcommand{\ba}{{\utwi{a}}}

\newcommand{\bb}{{\utwi{b}}}

\newcommand{\bff}{{\utwi{f}}}
\newcommand{\bg}{{\utwi{g}}}

\newcommand{\bq}{{\utwi{q}}}

\newcommand{\bs}{{\utwi{s}}}

\newcommand{\bu}{{\utwi{u}}}
\newcommand{\bv}{{\utwi{v}}}

\newcommand{\bx}{{\utwi{x}}}
\newcommand{\by}{{\utwi{y}}}
\newcommand{\bz}{{\utwi{z}}}
\newcommand{\bA}{{\utwi{A}}}
\newcommand{\bB}{{\utwi{B}}}
\newcommand{\bC}{{\utwi{C}}}

\newcommand{\bE}{{\utwi{E}}}

\newcommand{\bH}{{\utwi{H}}}
\newcommand{\bI}{{\utwi{I}}}

\newcommand{\bM}{{\utwi{M}}}
\newcommand{\bN}{{\utwi{N}}}

\newcommand{\bP}{{\utwi{P}}}
\newcommand{\bQ}{{\utwi{Q}}}
\newcommand{\bR}{{\utwi{R}}}

\newcommand{\bU}{{\utwi{U}}}
\newcommand{\bV}{{\utwi{V}}}
\newcommand{\bW}{{\utwi{W}}}
\newcommand{\bX}{{\utwi{X}}}
\newcommand{\bY}{{\utwi{Y}}}
\newcommand{\bZ}{{\utwi{Z}}}

\newcommand{\bPsi}{{\utwi{\mathnormal\Psi}}}
\newcommand{\bPhi}{{\utwi{\mathnormal\Phi}}}
\newcommand{\bXi}{{\utwi{\mathnormal\Xi}}}

\newcommand{\bxi}{{\utwi{\mathnormal\xi}}}

\newcommand{\bepsilon}{{\utwi{\epsilon}}}

\newcommand{\bOmega}{{\utwi{\Omega}}}
\renewcommand{\bPhi}{{\utwi{\Phi}}}
\renewcommand{\bPsi}{{\utwi{\Psi}}}
\newcommand{\bSigma}{{\utwi{\Sigma}}}

\newcommand{\bzero}{{\utwi{0}}}
\newcommand{\calA}{{\mathcal A}}

\newcommand{\calC}{{\mathcal C}}
\newcommand{\calD}{{\mathcal D}}

\newcommand{\calM}{{\mathcal M}}
\newcommand{\calN}{{\mathcal N}}
\newcommand{\calR}{{\mathcal R}}
\newcommand{\calS}{{\mathcal S}}

\newcommand{\calU}{{\mathcal U}}

\def\spacingset#1{\renewcommand{\baselinestretch}%
{#1}\small\normalsize} \spacingset{1}

\newtheorem{theorem}{Theorem}
\newtheorem{lemma}{Lemma} 

\newtheorem{proposition}{Proposition}
\newtheorem{remarkx}{Remark}

\newtheorem{condx}{Condition}
\newenvironment{condition}{\begin{condx} \em}{\end{condx}}

\DeclarePairedDelimiterX{\norm}[1]{\lVert}{\rVert}{#1}   

\newcommand{\vct}[1]{\mathbf{#1}}
\newcommand{\mtx}[1]{\mathbf{#1}}

\newcommand{\I}[1]{{\rm I \{ #1 \}}} 

\newcommand{\E}[1]{{\rm E} \left \{ #1 \right \}}    
\newcommand{\Var}[1]{{\rm Var} \{ #1 \}}
\newcommand{\Cov}[1]{{\rm Cov} \{ #1 \}}
\newcommand{\rank}[1]{{\rm rank} \left( #1 \right)}
\newcommand{\vect}[1]{{\rm vec} \left( #1 \right)}

\newcounter{question}

\renewcommand{\hat}{\widehat}
\renewcommand{\tilde}{\widetilde}

%% file: LCKproofs.tex

\section{Proofs}  

\subsection{Factor loadings}

\begin{lemma}  \label{lemma:entrywise_conv_rate_cov_vecXt}
Let $X_{t, ij}$ denote the $ij$-th entry of $\bX_t$. Under Condition \ref{cond:vecXt_alpha_mixing} and \ref{cond:Xt_cov_fullrank_bounded}, for any $i, k = 1, \ldots, d$ and $j, l = 1, \cdots, r$, we have 
\begin{equation}
\left|\frac{1}{T} \sum_{t=1}^{T} \left( X_{t, ij} X_{t,kl} - Cov(X_{t, ij} X_{t,kl}) \right) \right| = O_p(T^{-1/2}).
\end{equation} 
\end{lemma}

\begin{lemma}  \label{lemma:4_conv_rates}
Under Conditions 1-6, it holds that
\begin{eqnarray}
\sum_{i=1}^{p} \sum_{j=1}^{p} \norm{ \widehat{\bOmega}_{s_1 s_2, ij} - \bOmega_{s_1 s_2, ij} }^2_2 & = & O_p((n_1n_2)^{1-\delta} p^{2-2\gamma} T^{-1}),   \label{Lemma2-sig-converg},\\
\sum_{i=1}^{p} \sum_{j=1}^{p} \norm{ \widehat{\bOmega}_{s_1 e_2, ij} - \bOmega_{s_1 e_2, ij} }^2_2 & = & O_p(n_1^{2-\delta} p^{2-\gamma} T^{-1}),  \label{Lemma2-sig-err-converg}, \\
\sum_{i=1}^{p} \sum_{j=1}^{p} \norm{ \widehat{\bOmega}_{e_1s_2, ij} - \bOmega_{e_1s_2, ij} }^2_2 & = & O_p(n_2^{2-\delta} p^{2-\gamma} T^{-1}),  \label{Lemma2-err-sig-converg}, \\
\sum_{i=1}^{p} \sum_{j=1}^{p} \norm{ \widehat{\bOmega}_{e_1 e_2, ij} - \bOmega_{e_1 e_2, ij} }^2_2 & = & O_p(n_1 n_2 p^2 T^{-1}).  \label{Lemma2-err-converg}
\end{eqnarray} 
\end{lemma}

\begin{lemma}  \label{lemma:conv_rate_covYt}
Under Conditions 1-6, it holds that
\begin{equation}
\sum_{i=1}^{p} \sum_{j=1}^{p} \norm{ \widehat{\bOmega}_{ij} - \bOmega_{ij} }^2_2 = = O_p \left( n_1^{2-\delta} p^{2-\gamma} T^{-1} + n_2^{2-\delta} p^{2-\gamma} T^{-1} + n_1 n_2 p^2 T^{-1} \right) .
\end{equation} 
\end{lemma}

\begin{proof}

\begin{align*}
\hat{\bOmega}_{ij} & = \frac{1}{T} \sum_{t=1}^{T} \bY_{1 t, \cdot i} \bY'_{2 t, \cdot j} \\
& = \frac{1}{T} \sum_{t=1}^{T} \left( \bA_1 \bX_t B_{i \cdot} + E_{t, \cdot i} \right) \left( \bA_2 \bX_t B_{j \cdot} + E_{t, \cdot j} \right)' \\
& = \hat{\bOmega}_{s, ij} + \hat{\bOmega}_{se, ij} + \hat{\bOmega}_{es, ij} + \hat{\bOmega}_{e, ij}.
\end{align*}

\begin{align*}
\sum_{i=1}^{p} \sum_{j=1}^{p} \norm{\hat{\bOmega}_{ij} - \bOmega_{ij}}^2_2 & \le 4 \sum_{i=1}^{p} \sum_{j=1}^{p} \left( \norm{ \hat{\bOmega}_{s_1s_2, ij} - \bOmega_{s_1s_2, ij} }^2_2 + \norm{\hat{\bOmega}_{s_1e_2, ij} - \bOmega_{s_1e_2, ij}}^2_2 + \norm{\hat{\bOmega}_{e_1s_2, ij} - \bOmega_{e_1s_2, ij}}^2_2 + \norm{\hat{\bOmega}_{e_1e_2, ij} - \bOmega_{e_1e_2, ij}}^2_2 \right) \\
& = O_p(n_1^{2-\delta} p^{2-\gamma} T^{-1} + n_2^{2-\delta} p^{2-\gamma} T^{-1} + n_1 n_2 p^2 T^{-1}) 
\end{align*}
\end{proof}

\begin{lemma}  \label{lemma:conv_rate_Mhat}
Under Conditions 1-6 and $m_1 p_1^{-1+\delta_1} m_2 p_2^{-1+\delta_2} T^{-1/2} = o_p(1)$, it holds that
\begin{equation}
\norm{ \widehat{\bM}_1 - \bM_1 }_2 = O_p \left( n^{2-\delta} p^{2-\gamma} T^{-1/2} \right).
\end{equation} 
\end{lemma}

\begin{proof}
\begin{align*}
\sum_{i=1}^{p} \sum_{j=1}^{p} \norm{\bOmega_{ij}}^2_2 & = \sum_{i=1}^{p} \sum_{j=1}^{p} \norm{ \bA_1 \frac{1}{T} \sum_{t=1}^{T} \Cov{ \bX_t B_{i \cdot}, \bX_t B_{j \cdot} } \bA'_2}^2_2 \\
& \le \sum_{i=1}^{p} \sum_{j=1}^{p} \norm{\bA_1}^2_2 \norm{\bA_2}^2_2 \norm{\frac{1}{T} \sum_{t=1}^{T} \E{ \bX_t B_{i \cdot} B'_{j \cdot} \bX'_t  }}^2_2 \\
& \le \norm{\bA_1}^2_2 \norm{\bA_2}^2_2 \sum_{i=1}^{p} \sum_{j=1}^{p} \norm{\frac{1}{T} \sum_{t=1}^{T} \E{ \bX_t \otimes \bX_t} \vect{B_{i \cdot} B'_{j \cdot}} }^2_2 \\
& \le \norm{\bA_1}^2_2 \norm{\bA_2}^2_2 \sum_{i=1}^{p} \sum_{j=1}^{p} \norm{\frac{1}{T} \sum_{t=1}^{T} \E{ \bX_t \otimes \bX_t} }^2_2 \norm{ \vect{B_{i \cdot} B'_{j \cdot}} }^2_2 \\
& = \norm{\bA_1}^2_2 \norm{\bA_2}^2_2 \sum_{i=1}^{p} \sum_{j=1}^{p} \norm{\frac{1}{T} \sum_{t=1}^{T} \E{ \bX_t \otimes \bX_t} }^2_2 \norm{ B_{i \cdot} B'_{j \cdot} }^2_F \\
& \le \norm{\bA_1}^2_2 \norm{\bA_2}^2_2 \norm{\frac{1}{T} \sum_{t=1}^{T} \E{ \bX_t \otimes \bX_t} }^2_2 \sum_{i=1}^{p} \sum_{j=1}^{p} \norm{ B_{i \cdot}}^2_2 \norm{B'_{j \cdot}}^2_2 \\
& = \norm{\bA_1}^2_2 \norm{\bA_2}^2_2 \norm{\frac{1}{T} \sum_{t=1}^{T} \E{ \bX_t \otimes \bX_t} }^2_2 \norm{B}^4_F \\
& \le \norm{\bA_1}^2_2 \norm{\bA_2}^2_2 \norm{\frac{1}{T} \sum_{t=1}^{T} \E{ \bX_t \otimes \bX_t} }^2_2 \cdot r^2 \cdot \norm{B}^4_2  \\
& = O_p \left( (n_1n_2)^{1-\delta} p^{2-2\gamma} \right) 
\end{align*}

Then, 
\begin{eqnarray*}
\norm{\hat{\bM}_1 - \bM_1}_2 & = & \norm[\bigg]{\sum_{i=1}^{p} \sum_{j=1}^{p} \left( \hat{\bOmega}_{ij}\hat{\bOmega}'_{ij} - \bOmega_{ij} \bOmega'_{ij} \right)}_2 \\
& \le & \sum_{i=1}^{p} \sum_{j=1}^{p} \norm{\hat{\bOmega}_{ij} - \bOmega_{ij}}^2_2 + 2 \sum_{i=1}^{p} \sum_{j=1}^{p} \norm{\bOmega_{ij}}_2 \norm{\hat{\bOmega}_{ij} - \bOmega_{ij}}_2 \\
& \le & \sum_{i=1}^{p} \sum_{j=1}^{p} \norm{\hat{\bOmega}_{ij} - \bOmega_{ij}}^2_2 + 2 \left( \sum_{i=1}^{p} \sum_{j=1}^{p} \norm{\bOmega_{ij}}^2_2 \cdot \sum_{i=1}^{p} \sum_{j=1}^{p} \norm{\hat{\bOmega}_{ij} - \bOmega_{ij}}^2_2 \right)^{1/2} \\
& = & O_p((n_1^{2-\delta} p^{2-\gamma} + n_2^{2-\delta} p^{2-\gamma} + n_1 n_2 p^2) T^{-1}) \\
&   & + O_p \left( ( (n_1^{3-2\delta} n_2^{1-\delta} p^{4-3\gamma} + n_1^{1-\delta} n_2^{3-2\delta} p^{4-3\gamma} + n_1^{2-\delta} n_2^{2-\delta} p^{4-2\gamma} ) T^{-1} )^{1/2} \right).
\end{eqnarray*}

\end{proof}

\begin{lemma}  \label{conv_rate_sigval_M_1}
Under Condition \ref{cond:eigenval_cov_Et_bounded} and \ref{cond:Xt_cov_fullrank_bounded}, we have 
\begin{equation*}
\lambda_i(\bM_1) \asymp (n_1 n_2)^{1-\delta} p^{2-2\gamma}, \qquad i = 1, 2, \ldots, k_1,	
\end{equation*}
where $\lambda_i(\bM_1)$ denotes the $i$-th largest singular value of $\bM_1$. 
\end{lemma}

\begin{theorem} 
Under Condition \ref{cond:vecXt_alpha_mixing}-\ref{cond:B_factor_strength} and $n^{\delta} p^{\gamma} T^{-1/2} = o(1)$, we have 
\begin{equation}
\calD \left( \calM(\hat{\bA}_i), \calM(\bA_i) \right) = O_p( ( n_1 n_2^{\delta-1} p^{\gamma} + n_1^{\delta-1} n_2 p^{\gamma} + n_1^{\delta} n_2^{\delta} p^{2\gamma} ) T^{-1} )^{1/2}.
\end{equation}
If $n_1 \asymp n_2 \asymp n$, we have
\begin{equation}
\calD \left( \calM(\hat{\bA}_i), \calM(\bA_i) \right) = O_p(  n^{\delta} p^{\gamma} T^{-1/2} ).
\end{equation}
\end{theorem}

\begin{proof}
By Perturbation Theorem,
\begin{eqnarray}
\norm{\hat{\bA}_1 - \bA_1}_2 & \le & \frac{8}{\lambda_{min}(\bM_1)} \norm{\hat{\bM}_1 - \bM_1}_2 \nonumber \\
& = & O_p( ( n_1 n_2^{\delta-1} p^{\gamma} + n_1^{\delta-1} n_2 p^{\gamma} + n_1^{\delta} n_2^{\delta} p^{2\gamma} ) T^{-1} ) \nonumber \\
&  &  + O_p( ( n_1 n_2^{\delta-1} p^{\gamma} + n_1^{\delta-1} n_2 p^{\gamma} + n_1^{\delta} n_2^{\delta} p^{2\gamma} ) T^{-1} )^{1/2} \nonumber \\
& = & O_p( ( n_1 n_2^{\delta-1} p^{\gamma} + n_1^{\delta-1} n_2 p^{\gamma} + n_1^{\delta} n_2^{\delta} p^{2\gamma} ) T^{-1} )^{1/2}. \nonumber
\end{eqnarray}

If $n_1 \asymp n_2 \asymp n/2$, we have $\norm{\hat{\bA}_1 - \bA_1}_2 = O_p(  n^{\delta} p^{\gamma} T^{-1/2} )$.

If set $n_2=c$ fixed and $n_1=n-c$, we have $\norm{\hat{\bA}_1 - \bA_1}_2 = O_p(  ( n p^{-\gamma} + n^{\delta})^{1/2} p^{\gamma} T^{-1/2} )$.

We have the same result for $\norm{\hat{\bA}_2 - \bA_2}_2$. 
\end{proof}

\begin{theorem} 
This proposition considers the error bound of signal estimator as in (\ref{eqn:signal_est}) for each partition. Under $n^{\delta} p^{\gamma} T^{-1}=o_p(1)$, if $n_1 \asymp n2 \asymp n$, then
\begin{equation}
n^{-1/2} p^{-1/2} \norm{\hat{\bXi}_{it} - \bXi_{it}}_2  = O_p(n^{\delta/2} p^{\gamma/2} T^{-1/2} + n^{-1/2} p^{-1/2}),
\end{equation}
for $i = 1, 2$, and 
\begin{equation}
n^{-1} p^{-1} \norm{\hat{\bXi}_{t} - \bXi_{t}}^2_2  = O_p(n^{\delta} p^{\gamma} T^{-1} + n^{-1/2+\delta/2} p^{-1/2+\gamma/2} T^{-1/2} + n^{-1} p^{-1})
\end{equation}
\end{theorem}

\begin{proof}
 
\begin{eqnarray*}
\norm{\hat{\bXi}_{1t} - \bXi_{1t}}_2  & = &  \norm[\Big]{ \hat{\bA}^{(j)}_1 \hat{\bA}^{(j)'}_1 \left( \bA^{(j)}_1 \bX_t \bB' + \bE_t^{(j)} \right) \hat{\bB}^{(j)} \hat{\bB}^{(j)'} - \bA^{(j)}_1 \bX_t \bB'}_2 \\
& \le &  \norm[\Big]{ \hat{\bA}^{(j)}_1 \hat{\bA}^{(j)'}_1 \bA^{(j)}_1 \bX_t \bB' \left( \hat{\bB}^{(j)} \hat{\bB}^{(j)'} - \bB \bB' \right) }_2 \\
& & +  \norm[\Big]{ \left( \hat{\bA}^{(j)}_1 \hat{\bA}^{(j)'}_1 - \bA^{(j)}_1 \bA^{(j)'}_1 \right) \bA^{(j)}_1 \bX_t \bB'}_2 \\
& & +  \norm[\Big]{ \hat{\bA}^{(j)}_1 \hat{\bA}^{(j)'}_1 \bE_t^{(j)} \hat{\bB}^{(j)} \hat{\bB}^{(j)'} }_2 \\
& = & \bI_1 + \bI_2 + \bI_3. 
\end{eqnarray*}

\begin{eqnarray*}
& \norm{\bI_1}_2 & \le 2 \norm{\bX_t}_2 \norm{\hat{\bB}^{(j)} - \bB^{(j)}}_2  = O_p(n_1^{1/2-\delta/2} p^{1/2-\gamma/2} \norm{\hat{\bB}^{(j)} - \bB}_2 ) \\
& & = O_p(n_1^{1/2-\delta/2} n^{\delta} p^{1/2+\gamma/2} T^{-1/2} ) \\
&\norm{\bI_2}_2 & \le 2 \norm{\hat{\bA}^{(j)}_1 - \bA^{(j)}_1}_2 \norm{\bX_t}_2 = O_p(n_1^{1/2-\delta/2} p^{1/2-\gamma/2} \norm{\hat{\bA}^{(j)}_1 - \bA^{(j)}_1}_2)\\
& & = O_p( ( n_1 n_2^{\delta-1} p^{\gamma} + n_1^{\delta-1} n_2 p^{\gamma} + n_1^{\delta} n_2^{\delta} p^{2\gamma} ) T^{-1} )^{1/2} n_1^{1/2-\delta/2} p^{1/2-\gamma/2} ) \\
& & = O_p((n_1^{2-\delta} n_2^{\delta-1} + n_2 + n_1 n_2^{\delta} p^{\gamma}) p T^{-1})^{1/2} \\
&\norm{\bI_3}_2 & \le \norm{\hat{\bA}^{(j)'}_1 \bE_t^{(j)} \hat{\bB}^{(j)}} = \norm{(\hat{\bB}^{(j)'} \otimes \hat{\bA}^{(j)'}_1 )\vect{\bE_t^{(j)}}}_2 \le dr \norm{\bSigma_e}_2 = O_p(1).
\end{eqnarray*}

Thus, 
\begin{equation}
\norm{\hat{\bXi}_{1t} - \bXi_{1t}}_2  = O_p(n_1^{1/2-\delta/2} p^{1/2-\gamma/2} \norm{\hat{\bA}^{(j)}_1 - \bA^{(j)}_1}_2) + O_p(n_1^{1/2-\delta/2} p^{1/2-\gamma/2} \norm{\hat{\bB}^{(j)} - \bB}_2 ) + O_p(1).
\end{equation}

\begin{equation}
n_1^{-1/2} p^{-1/2} \norm{\hat{\bXi}_{1t} - \bXi_{1t}}_2  = O_p(n_1^{-\delta/2} p^{-\gamma/2} \norm{\hat{\bA}^{(j)}_1 - \bA^{(j)}_1}_2) + O_p(n_1^{-\delta/2} p^{-\gamma/2} \norm{\hat{\bB}^{(j)} - \bB}_2 ) + O_p(n_1^{-1/2} p^{-1/2}).
\end{equation}

Similarly for ${\bXi}_{2t}$, we have 
\begin{equation}
\norm{\hat{\bXi}_{2t} - \bXi_{2t}}_2  = O_p(n_2^{1/2-\delta/2} p^{1/2-\gamma/2} \norm{\hat{\bA}^{(j)}_2 - \bA^{(j)}_2}_2) + O_p(n_2^{1/2-\delta/2} p^{1/2-\gamma/2} \norm{\hat{\bB}^{(j)} - \bB}_2 ) + O_p(1).
\end{equation}

\begin{equation}
n_2^{-1/2} p^{-1/2} \norm{\hat{\bXi}_{2t} - \bXi_{2t}}_2  = O_p(n_2^{-\delta/2} p^{-\gamma/2} \norm{\hat{\bA}^{(j)}_1 - \bA^{(j)}_1}_2) + O_p(n_2^{-\delta/2} p^{-\gamma/2} \norm{\hat{\bB}^{(j)} - \bB}_2 ) + O_p(n_2^{-1/2} p^{-1/2}).
\end{equation}

\textcolor{red}{If $n_1 \asymp n_2 \asymp n$, then}
\begin{equation}
\norm{\hat{\bXi}_{it} - \bXi_{it}}_2  = O_p(n^{1/2+\delta/2} p^{1/2+\gamma/2} T^{-1/2}) + O_p(1), \qquad i = 1,2.
\end{equation}

Now we find the $L_2$-norm bounds for 
\begin{equation*}
\norm{\hat{\bXi}_{t} - \bXi_{t}}^2_2 = \norm[\bigg]{ \begin{pmatrix}
\hat{\bXi}_{1t} - \bXi_{1t} \\
\hat{\bXi}_{2t} - \bXi_{2t}
\end{pmatrix}  }^2_2 .
\end{equation*}

Let $\bM = \hat{\bXi}_{t} - \bXi_{t} = \begin{pmatrix} \bM_1 \\ \bM_2 \end{pmatrix}$, the above problem is equivelent to finding $\lambda_{max}(\bM'\bM)$ from $\lambda_{max}(\bM_1'\bM_1)$ and $\lambda_{max}(\bM_2'\bM_2)$. \\

Since 
\begin{equation*}
\lambda_{max}(\bM'\bM) = \lambda_{max}(\bM_1'\bM_1 + \bM_2'\bM_2) \le \lambda_{max}(\bM_1'\bM_1) + \lambda_{max}(\bM_2'\bM_2), 
\end{equation*}

We have 
\begin{align*}
\norm{\hat{\bXi}_{t} - \bXi_{t}}^2_2  & \le \norm{\hat{\bXi}_{1t} - \bXi_{1t}}^2_2  + \norm{\hat{\bXi}_{2t} - \bXi_{2t}}^2_2  \\
& = O_p(n^{1+\delta} p^{1+\gamma} T^{-1}) + O_p(n^{1/2+\delta/2} p^{1/2+\gamma/2} T^{-1/2}) + O_p(1).
\end{align*}

$n^{-1} p^{-1} \norm{\hat{\bXi}_{t} - \bXi_{t}}^2_2  = O_p(n^{\delta} p^{\gamma} T^{-1} + n^{-1/2+\delta/2} p^{-1/2+\gamma/2} T^{-1/2} + n^{-1} p^{-1})$.
\end{proof}

\subsection{Space factor loading matrix re-estimation}

\begin{lemma}
If $n_1 \asymp n2 \asymp n$, then
\begin{equation}
n^{-1/2} p^{-1/2} \norm{\hat{\bPsi}_{it} - \bPsi_{it}}_2  = O_p(n^{\delta/2} p^{\gamma/2} T^{-1/2}) + O_p(n^{-1/2} p^{-1/2}),
\end{equation}
for $i = 1, 2$, and 
\begin{equation}
n^{-1} p^{-1} \norm{\hat{\bPsi}_{t} - \bPsi_{t}}^2_2  = O_p(n^{\delta} p^{\gamma} T^{-1} + n^{-1/2+\delta/2} p^{-1/2+\gamma/2} T^{-1/2} + n^{-1} p^{-1})
\end{equation}
\end{lemma}

\begin{proof} 

\begin{eqnarray*}
\norm{\bPsi_{it} - \bPsi_{it}}_2 & = & \norm[\big]{ \hat{\bQ}_{A_i} \hat{\bZ}_t - \bQ_{A_i} \bZ_t}_2 = \norm[\big]{ \hat{\bQ}_{A_i} \hat{\bQ}'_{A_i} (\bQ_{A_i} \bZ_t \bQ'_B + \bE_t) \hat{\bQ}_B - \bQ_{A_i} \bZ_t}_2 \\
& = & \norm[\big]{ \hat{\bQ}_{A_i} \hat{\bQ}'_{A_i} \bQ_{A_i} \bZ_t \bQ'_B ( \hat{\bQ}_B - \bQ_B) + (\hat{\bQ}_{A_i} \hat{\bQ}'_{A_i} - \bQ_{A_i} \bQ'_{A_i})\bQ_{A_i} \bZ_t + \hat{\bQ}_{A_i} \hat{\bQ}'_{A_i} \bE_t \hat{\bQ}_B}_2 \\
& \le & \norm[\big]{ \hat{\bQ}_{A_i} \hat{\bQ}'_{A_i} \bQ_{A_i} \bZ_t \bQ'_B ( \hat{\bQ}_B - \bQ_B) }_2 + \norm[\big]{ (\hat{\bQ}_{A_i} \hat{\bQ}'_{A_i} - \bQ_{A_i} \bQ'_{A_i})\bQ_{A_i} \bZ_t }_2 + \norm[\big]{ \hat{\bQ}_{A_i} \hat{\bQ}'_{A_i} \bE_t \hat{\bQ}_B}_2
\end{eqnarray*}

Then, similar to the proof of Theorem \ref{prop:signal_error_bound}, we have the desired results. 
\end{proof}

Let $\bU_t = \hat{\bPsi}_{t} - \bPsi_{t}$ and $\Delta_{npT} = n^{\delta} p^{\gamma} T^{-1} + n^{-1/2+\delta/2} p^{-1/2+\gamma/2} T^{-1/2} + n^{-1} p^{-1}$. Then $\Delta_{npT}$ is the convergence rate of $n^{-1} p^{-1} \norm{\bU_t}^2_2$. Since $\norm{\bU_t}^2_2 \le \norm{\bU_t}^2_F \le r \norm{\bU_t}^2_2$ where $r$ is fixed, we have $n^{-1} p^{-1} \norm{\bU_t}^2_F = O_p(\Delta_{npT})$.

Define $\bW_t = \bX_t \bR'_B$, $\bW = (\bW_1 \cdots \bW_T)$, $\bPsi = (\bPsi_1 \cdots \bPsi_T) = \bA \bW $. Assume $\frac{1}{rT} \bW \bW' = \bI_d$. The the columns of $\bW$ compose of the eigenvectors of $\frac{1}{nprT}\bPsi' \bPsi = \frac{1}{nprT} \bW' \bA' \bA \bW$ corresponding to the $d$ nonzero eigenvalues. However, we only have the estimate of $\hat{\bPsi} = (\hat{\bPsi}_1 \cdots \hat{\bPsi}_T)$. Thus, $\hat{\bW}$ and $\hat{\bA}$ can be estimated from $\frac{1}{nprT}\hat{\bPsi}' \hat{\bPsi} = \frac{1}{nprT}(\bPsi+\bU)'(\bPsi+\bU) $, where $\bU = (\bU_1 \cdots \bU_T)$ is the approximation error from the previous steps.

Let $\bV_{npT}$ be the $d \times d$ diagonal matrix of the first $d$ largest eigenvalues of $\frac{1}{nprT} \hat{\bPsi}'\hat{\bPsi}$ in decreasing order. By definition of eigenvectors and eigenvalues, we have $\frac{1}{nprT} \hat{\bPsi}' \hat{\bPsi} \hat{\bW}' = \hat{\bW}' \bV_{npT}$ or $\frac{1}{nprT} \hat{\bPsi}' \hat{\bPsi} \hat{\bW}' \bV_{npT}^{-1} = \hat{\bW}'$.

Define $\bH = \frac{1}{nprT} \bA'\bA \bW \hat{\bW}' \bV^{-1}_{npT}$, then 
\begin{eqnarray*}
\hat{\bW}' - \bW'\bH & = & \frac{1}{nprT} \hat{\bPsi}' \hat{\bPsi} \hat{\bW}' \bV_{npT}^{-1} - \frac{1}{nprT} \bW' \bA'\bA \bW \hat{\bW}' \bV^{-1}_{npT} \\
& = & \left( \frac{1}{nprT} \bW' \bA' \bU \hat{\bW}' + \frac{1}{npT} \bU' \bA \bW \hat{\bW}' + \frac{1}{nprT} \bU'\bU \hat{\bW}' \right) \bV^{-1}_{npT} \\
& = & \left( \bN_1 + \bN_2 + \bN_3 \right) \bV^{-1}_{npT}.
\end{eqnarray*} 

\begin{lemma}
$\frac{1}{rT} \norm{\bN_1}^2_F = \frac{1}{rT} \norm{\bN_2}^2_F = O_p(n^{-\delta} p^{-\gamma} \Delta_{npT})$ and $\frac{1}{rT} \norm{\bN_3}^2_F = O_p(\Delta_{npT}^2)$.
\end{lemma}

\begin{proof}

Note that $\norm{\bU}^2_F = \norm{\hat{\bPsi} - \bPsi}^2_F = \norm[\big]{\sum_{t=1}^{T}(\hat{\bPsi}_t - \bPsi_t)}^2_F \le T \; \underset{1 \le t \le T}{\max} \norm{\hat{\bPsi}_t - \bPsi_t}^2_F = O_p(npT \Delta_{npT} )$ and $\norm{\bW}^2_F = \norm{\hat{\bW}}^2_F = O_p(rT)$ and $r$ is fixed. In addition, we have $\norm{\bA}^2_F \asymp \norm{\bA}^2_2 = O_p(n^{1-\delta}p^{1-\gamma})$.

Thus, 
\begin{eqnarray*}
\frac{1}{rT} \norm{\bN_1}^2_F & \le & \frac{1}{n^2p^2r^3T^3} \norm{\bW}^2_F \norm{\bA}^2_F \norm{\bU}^2_F \norm{\hat{\bW}}^2_F = O_p(n^{-\delta} p^{-\gamma} \Delta_{npT}) \\
\frac{1}{rT} \norm{\bN_2}^2_F & \le & \frac{1}{n^2p^2r^3T^3} \norm{\bU}^2_F \norm{\bA}^2_F \norm{\bW}^2_F \norm{\hat{\bW}}^2_F = O_p(n^{-\delta} p^{-\gamma} \Delta_{npT}) \\
\frac{1}{rT} \norm{\bN_3}^2_F & \le & \frac{1}{n^2p^2r^3T^3} \norm{\bU}^4_2 \norm{\hat{\bW}}^2_F = O_p( \Delta^2_{npT} )
\end{eqnarray*}
\end{proof}

\begin{lemma}
(i) $\norm{\bV_{npT}}_2 = O_p(n^{-\delta} p^{-\gamma})$, $\norm{\bV_{npT}^{-1}}_2 = O_p(n^{\delta} p^{\gamma})$. \\
(ii) $\norm{\bH}_2 = O_p(1)$. 
\end{lemma}

\begin{proof}
The $d$ eigenvalues of $\bV_{npT}$ are the same as those of $\frac{1}{nprT} \hat{\bPsi}\hat{\bPsi}' = \frac{1}{np} \bA \bA' + \frac{1}{nprT} \bA \bW \bU' + \frac{1}{nprT} \bU \bW' \bA' + \frac{1}{nprT} \bU \bU'$, which follows from $\hat{\bPsi}=\bA \bW + \bU$ and $\bW \bW' / rT = I_d$. Thus

\begin{equation*}
\norm{\frac{1}{nprT} \hat{\bPsi}\hat{\bPsi}' - \frac{1}{np} \bA \bA'}_2 \le \frac{1}{nprT} \norm{ \bA \bW \bU_t'}_2 + \frac{1}{nprT} \norm{ \bU \bW' \bA'}_2 + \frac{1}{nprT} \norm{\bU \bU_t'} = o_p(1).
\end{equation*}

Using the inequality for the $k$th eigenvalue, $|\lambda_k(\bW) - \lambda_k(\bW_1)| \le |\bW - \bW_1|$, we have $|\lambda_k(\frac{1}{nprT} \hat{\bPsi}\hat{\bPsi}') - \lambda(\frac{1}{np} \bA \bA')| = o_p(1)$. $\lambda_k(\frac{1}{np} \bA \bA') \asymp n^{-\delta} p^{-\gamma}$, $k=1, \ldots, d$. Thus, $\norm{\bV_{npT}}_{min} \asymp n^{-\delta} p^{-\gamma} \asymp \norm{\bV_{npT}}_2$, $\norm{\bV_{npT}^{-1}}_{min} \asymp n^{\delta} p^{\gamma} \asymp \norm{\bV_{npT}^{-1}}_2$, and $\norm{\bH}_2 = O_p(1)$.

\end{proof}
  
\begin{lemma}
\begin{equation*}
\frac{1}{rT}\norm{\hat{\bW}' - \bW'\bH}^2_F = O_p(\Delta_{npT} + n^{\delta} p^{\gamma} \Delta_{npT}^2) 
\end{equation*}
\end{lemma}

\begin{proof}
Follow from Lemma 6, 7 and 8. 
\end{proof}

\begin{lemma}
\[
\norm{\bH - \bI_d}_F =  O_p \left(\Delta_{npT} + n^{\delta} p^{\gamma}\Delta_{npT}^2 \right)  + O_p \left(\Delta_{npT} T^{-1} + n^{\delta} p^{\gamma}\Delta_{npT}^2 T^{-1} \right)^{1/2}.
\]
\end{lemma}

\begin{proof}
$\bH = \frac{1}{nprT} \bA'\bA \bW \hat{\bW}' \bV^{-1}_{npT}$

\begin{eqnarray*}
\norm[\bigg]{\bI_d - \frac{1}{rT}\hat{\bW} \bW'\bH}_F & = & \norm[\bigg]{\frac{1}{rT}\hat{\bW}(\hat{\bW}' - \bW'\bH)}_F \\
& \le & \frac{1}{rT} \norm{\hat{\bW}' - \bW'\bH}^2_F + \frac{1}{rT}\norm{\bW(\hat{\bW}' - \bW' \bH)}_F \\
& = & O_p \left(\Delta_{npT} + n^{\delta} p^{\gamma}\Delta_{npT}^2\right)  + O_p \left(\Delta_{npT} T^{-1} + n^{\delta} p^{\gamma}\Delta_{npT}^2 T^{-1} \right) ^{1/2}
\end{eqnarray*}

\begin{eqnarray*}
\norm[\bigg]{\frac{1}{rT} \hat{\bW} \bW' \bH - \bH'\bH }_F & = & \norm[\bigg]{ \frac{1}{rT}(\hat{\bW}' - \bW'H)'\bW'\bH }_F = O_p \left(\Delta_{npT} T^{-1} + n^{\delta} p^{\gamma}\Delta_{npT}^2 T^{-1} \right) ^{1/2}
\end{eqnarray*}

Thus, 
\[
\norm[\bigg]{\bI_d - \bH'\bH}_F = O_p \left(\Delta_{npT} + n^{\delta} p^{\gamma}\Delta_{npT}^2\right)  + O_p \left(\Delta_{npT} T^{-1} + n^{\delta} p^{\gamma}\Delta_{npT}^2 T^{-1} \right) ^{1/2}
\]

In addition, by the definition of $\bH = \frac{1}{nprT} \bA'\bA \bW \hat{\bW}' \bV^{-1}_{npT}$, we have
\[
\norm[bigg]{\bH \bV_{npT} - \frac{1}{np} \bA'\bA \bH }_F = \frac{1}{nprT} \bA'\bA \bW (\hat{\bW}' - \bW' \bH) = O_p \left( n^{-\delta} p^{-\gamma} (\Delta_{npT} + n^{\delta} p^{\gamma}\Delta_{npT}^2)^{-1/2} \right).
\]

With the same argument of Proposition C.3 in \cite{fan2016projected}, we have 
\[
\norm{\bH - \bI_d}_F =  O_p \left(\Delta_{npT} + n^{\delta} p^{\gamma}\Delta_{npT}^2 \right)  + O_p \left(\Delta_{npT} T^{-1} + n^{\delta} p^{\gamma}\Delta_{npT}^2 T^{-1} \right)^{1/2}.
\]

\end{proof}

\begin{theorem}
\[
\frac{1}{rT} \norm{\hat{\bW}' - \bW'}^2_F = O_p \left(\Delta_{npT} + n^{\delta} p^{\gamma}\Delta_{npT}^2\right)
\]
\end{theorem}
\begin{proof}
\[
\frac{1}{rT} \norm{\hat{\bW}' - \bW'}^2_F \le \frac{2}{rT} \norm{\hat{\bW}' - \bW'\bH}^2_F + 2 \norm{\bH - \bI_d}^2_F = O_p \left(\Delta_{npT} + n^{\delta} p^{\gamma}\Delta_{npT}^2\right)
\]
\end{proof}

\begin{proposition}
\[
\frac{1}{np} \norm[\big]{\hat{\bA} - \bA}^2_F = O_p \left( \Delta_{npT} \right). 
\]
\end{proposition}

\begin{proof}
\[
\hat{\bA} = \frac{1}{rT} \hat{\bPsi} \hat{\bW}'
\]
\begin{equation*}
\frac{1}{rT} \norm{\bPsi}^2_2 = \norm{\bPsi\bPsi'}_2 = \norm[\bigg]{ \frac{1}{rT} \sum_{t=1}^{T} \bPsi_t\bPsi'_t }_2 \le \underset{1\le t \le T}{\max} \norm{ \bPsi_t\bPsi'_t }_2 / r = O_p(n^{1-\delta}p^{1-\gamma})
\end{equation*}

\begin{eqnarray*}
\frac{1}{np} \norm[\big]{\hat{\bA} - \bA}^2_F & = & \frac{1}{np} \norm[\bigg]{\frac{1}{rT}\hat{\bPsi} \hat{\bW}' - \frac{1}{rT}{\bPsi}{\bW}'}^2_F = \frac{1}{np} \norm[\bigg]{\frac{1}{rT}(\hat{\bPsi}-{\bPsi}) \hat{\bW}' + \frac{1}{rT} {\bPsi} (\hat{\bW}' - {\bW}')}^2_F \\
& \le & 2 \frac{\norm{\bU}^2_F}{nprT} \cdot \frac{1}{rT} \norm{\hat{\bW}'}^2_F + 2 \frac{\norm{\bPsi}^2_F}{nprT} \cdot \frac{1}{rT} \norm{\hat{\bW}' - \bW'}^2_F \\
& = & O_p \left( \Delta_{npT} + n^{-\delta} p^{-\gamma}  (\Delta_{npT} + n^{\delta} p^{\gamma}\Delta_{npT}^2 ) \right) \\
& = & O_p ( \Delta_{npT} )
\end{eqnarray*} 
\end{proof}

\subsection{Sieve approximation of space loading function}

$\bA(\bs) = (a_1(\bs), \cdots, a_d(\bs))$, now we want to approximate $a_j(\bs)$ with linear combination of basis functions, the approximating functions are $\hat{g}_j(\bs)$. We estimate $\hat{g}_j(\bs)$ based on estimated value $\hat{A}_{\cdot j}$'s. $\hat{A}_{\cdot j} = {A}_{\cdot j} + {E}_{A,\cdot j}$. Since for $n \times d$ matrix $\bA$ with fixed column dimension $d$, $\norm{\bA}^2_2 \le \norm{\bA}^2_F \le d \norm{\bA}^2_2$. we have $\norm{\hat{A}_{\cdot j} - {A}_{\cdot j}}^2_2 = O_p(np\Delta_{npT})$, $j = 1, \ldots, d$. 

${A}_{\cdot j} = a_j(\bs)$, then $\hat{A}_{\cdot j} = \hat{a}_j(\bs)= a_j(\bs) + e_{a,j}(\bs)$. 

\begin{lemma}
If H\"{o}lder class, then $| a_{j}(\bs) |^2_{\infty} \asymp n^{-\delta} p^{1-\gamma}$, $| e_{a,j}(\bs) |^2_{\infty} = O_p(p\Delta_{npT})$.
\end{lemma}

\begin{proof}
\begin{eqnarray*}
\lambda_{max}(\bA \bA') =  \lambda_{max}(\sum_{j=1}^{d} A_{\cdot j} A'_{\cdot j}) \ge \lambda_{min}(\sum_{j=1}^{d} A_{\cdot j} A'_{\cdot j}) & \ge & \sum_{j=1}^{d} \lambda_{min}(A'_{\cdot j}A_{\cdot j}) = \sum_{j=1}^{d} \sum_{i=1}^{n} A_{i j}^2 \\
\lambda_{min}(\bA \bA') = \lambda_{min}(\sum_{j=1}^{d} A_{\cdot j} A'_{\cdot j}) \le \lambda_{max}(\sum_{j=1}^{d} A_{\cdot j} A'_{\cdot j}) & \le & \sum_{j=1}^{d} \lambda_{max}(A'_{\cdot j}A_{\cdot j}) = \sum_{j=1}^{d} \sum_{i=1}^{n} A_{i j}^2 
\end{eqnarray*}

Since $\norm{\bA}^2_{min} \asymp \norm{\bA}^2_{max} \asymp n^{1-\delta}p^{1-\gamma}$, then $\| A_{\cdot j} \|^2 \asymp n^{1-\delta} p^{1-\gamma}$.

If H\"{o}lder class, then $| a_{j}(\bs) |^2_{\infty} \asymp n^{-\delta} p^{1-\gamma}$ by multivariate Taylor expansion and Sandwich Theorem.

\end{proof}

\begin{lemma}
$\norm{\hat{g}_j(\bs) - a_j(\bs)}_{\infty} = O_p(J_n^{-\kappa} n^{-\delta/2} p^{1/2-\gamma/2}) + O_p(\sqrt{p\Delta_{npT}})$.
\end{lemma}

\begin{proof}
Following Theorem 12.6, 12.7 and 12.8 in \cite{schumaker2007spline}, we have $\norm{\hat{g}_j(\bs) - a_j(\bs)}_{\infty} = \norm{\bP \hat{a}_j(\bs) - a_j(\bs)} \le \norm{\bP a(\bs) - a_j(\bs)} + \norm{ \bP e_{a,j}(\bs) - e_{a,j}(\bs) } + \norm{ e_{a,j}(\bs) } = O_p(J_n^{-\kappa} n^{-\delta/2} p^{1/2-\gamma/2}) + O_p(\sqrt{p\Delta_{npT}})$.
\end{proof}

\begin{theorem}
\begin{equation}
\frac{1}{pT} \norm{\hat{\bxi}(\bs_0) - \bxi(\bs_0)}^2_2 = O_p(J_n^{-2\kappa} n^{-\delta} p^{-\gamma} + \Delta_{npT} + 1/T)
\end{equation}
\end{theorem}

\begin{proof}
Let $\bxi'_t(\bs_0) = \ba'(\bs_0) \bX_t \bB = \ba'(\bs_0) \bX_t \bR_{B}' \bQ_{B}'$

$\bxi'(\bs_0) = (\bxi'_1(\bs_0) \cdots \bxi'_T(\bs_0)) = (\ba'(\bs_0) \bX_1 \bR_{B}' \bQ_{B}' \cdots \ba'(\bs_0) \bX_T \bR_{B}' \bQ_{B}') = \ba'(\bs_0) \bW (\bI_T \otimes \bQ_B')$.

$\hat{\bxi}'(\bs_0) = \hat{\bg}'(\bs_0) \hat{\bW} (\bI_T \otimes \hat{\bQ}_B')$.

$\hat{\bxi}'(\bs_0) - \bxi'(\bs_0) = \hat{\bg}'(\bs_0) \hat{\bW}' (\bI_T \otimes \hat{\bQ}_B') - \ba'(\bs_0) \bW' (\bI_T \otimes \bQ_B')$.

\begin{eqnarray*}
\hat{\bxi}(\bs_0) - \bxi(\bs_0) & = & (\bI_T \otimes \hat{\bQ}_B) \hat{\bW}' \hat{\bg}(\bs_0)   - (\bI_T \otimes \bQ_B) \bQ_W \ba(\bs_0) \\
& = & (\bI_T \otimes \hat{\bQ}_B) \left( \hat{\bW}' \hat{\bg}(\bs_0) - \bQ_W \ba(\bs_0) \right) + \left( \bI_T \otimes (\hat{\bQ}_B - {\bQ}_B) \right) \bQ_W \ba(\bs_0) \\
& = & (\bI_T \otimes \hat{\bQ}_B) \hat{\bW}' \left( \hat{\bg}(\bs_0) - \ba(\bs_0) \right) + (\bI_T \otimes \hat{\bQ}_B)\left( \hat{\bW}' - \bW' \right) \ba(\bs_0) + \left( \bI_T \otimes (\hat{\bQ}_B - {\bQ}_B) \right) \bW' \ba(\bs_0)  
\end{eqnarray*}

\begin{eqnarray*}
\frac{1}{\sqrt{T}} \norm{(\bI_T \otimes \hat{\bQ}_B) \hat{\bW}' \left( \hat{\bg}(\bs_0) - \ba(\bs_0) \right)}_2 & \le & \frac{1}{\sqrt{T}} \norm{\bI_T \otimes \hat{\bQ}_B}_2 \norm{\hat{\bW}'}_2 \norm{\hat{\bg}(\bs_0) - \ba(\bs_0)}_2 \\
& = & O_p(J_n^{-\kappa} n^{-\delta/2} p^{1/2-\gamma/2} + \sqrt{p\Delta_{npT}}) \\
\frac{1}{T} \norm{(\bI_T \otimes \hat{\bQ}_B)\left( \hat{\bW}' - \bW' \right) \ba(\bs_0)}^2_2 & \le & \frac{1}{T} \norm{\bI_T \otimes \hat{\bQ}_B}^2_2 \norm{\hat{\bW}' - \bW'}_2 \norm{\ba(\bs_0)}^2_2 \\
& = & O_p\left(\Delta_{npT} + n^{\delta} p^{\gamma}\Delta_{npT}^2\right) O_p(n^{-\delta}p^{1-\gamma}) \\
& = & O_p( n^{-\delta}p^{1-\gamma}\Delta_{npT} + p \Delta_{npT}^2  ) \\
\frac{1}{\sqrt{T}} \norm{\left( \bI_T \otimes (\hat{\bQ}_B - {\bQ}_B) \right) \bW' \ba(\bs_0)}_2 & \le & \frac{1}{\sqrt{T}} \norm{\bI_T \otimes (\hat{\bQ}_B - {\bQ}_B)}_2 \norm{\bW'}_2 \norm{\ba(\bs_0)}_2 \\
& = & O_p(n^{\delta/2}p^{\gamma/2}T^{-1/2})O_p(n^{-\delta/2}p^{1/2-\gamma/2}) \\
& = & O_p(\sqrt{p/T})
\end{eqnarray*}

Thus, 
\begin{equation}
\frac{1}{pT} \norm{\hat{\bxi}(\bs_0) - \bxi(\bs_0)}^2_2 = O_p(J_n^{-2\kappa} n^{-\delta} p^{-\gamma} + \Delta_{npT} + 1/T).
\end{equation}

\end{proof}

%% file: LCKmoretableplots.tex
\section{Tables and Plots}  \label{appendix:tableplots}

\begin{table}[htpb!]
\centering
\caption{Mean and standard deviations (in parentheses) of the estimated accuracy measured by $\calD(\hat{\cdot}, \cdot)$ for spatial and variable loading matrices. All numbers in the table are 10 times the true numbers for clear representation. The results are based on 200 simulations.}
\label{table:spdist_msd_table}
\resizebox{\textwidth}{!}{%
\begin{tabular}{ccc|ccccc|ccccl}
\hline
\multicolumn{3}{c|}{} & \multicolumn{5}{c|}{$\gamma = 0$} & \multicolumn{5}{c}{$\gamma=0.5$} \\ \hline
T & p & n & $\calD(\hat{\bA}_1, \bA_1)$ & $\calD(\hat{\bA}_2, \bA_2)$ & Average & $\calD(\hat{\bA}, \bA)$ & $\calD(\hat{\bB}, \bB)$ & $\calD(\hat{\bA}_1, \bA_1)$ & $\calD(\hat{\bA}_2, \bA_2)$ & Average & $\calD(\hat{\bA}, \bA)$ & $\calD(\hat{\bB}, \bB)$ \\ \hline
60 & 10 & 50 & 0.68(0.1) & 0.67(0.1) & 0.68(0.08) & 0.67(0.07) & 0.53(0.11) & 1.27(0.19) & 1.25(0.21) & 1.26(0.16) & 1.25(0.15) & 0.69(0.14) \\
120 & 10 & 50 & 0.45(0.06) & 0.46(0.06) & 0.45(0.05) & 0.45(0.04) & 0.5(0.12) & 0.83(0.12) & 0.84(0.12) & 0.84(0.09) & 0.84(0.08) & 0.63(0.13) \\
240 & 10 & 50 & 0.31(0.04) & 0.31(0.04) & 0.31(0.03) & 0.31(0.02) & 0.49(0.11) & 0.57(0.07) & 0.57(0.08) & 0.57(0.05) & 0.57(0.04) & 0.6(0.13) \\ \hdashline
60 & 20 & 50 & 0.5(0.07) & 0.5(0.09) & 0.5(0.06) & 0.5(0.06) & 0.52(0.08) & 1.18(0.21) & 1.18(0.24) & 1.18(0.17) & 1.17(0.15) & 0.69(0.1) \\
120 & 20 & 50 & 0.34(0.05) & 0.34(0.05) & 0.34(0.03) & 0.34(0.03) & 0.5(0.07) & 0.79(0.12) & 0.79(0.12) & 0.79(0.09) & 0.78(0.08) & 0.6(0.08) \\
240 & 20 & 50 & 0.23(0.03) & 0.23(0.03) & 0.23(0.02) & 0.23(0.02) & 0.47(0.06) & 0.52(0.07) & 0.52(0.07) & 0.52(0.05) & 0.52(0.05) & 0.54(0.06) \\ \hdashline
60 & 40 & 50 & 0.32(0.06) & 0.32(0.05) & 0.32(0.04) & 0.32(0.04) & 0.49(0.07) & 0.98(0.21) & 0.95(0.19) & 0.96(0.15) & 0.95(0.13) & 0.67(0.07) \\
120 & 40 & 50 & 0.21(0.03) & 0.21(0.03) & 0.21(0.02) & 0.21(0.02) & 0.48(0.05) & 0.63(0.1) & 0.62(0.1) & 0.63(0.08) & 0.62(0.07) & 0.58(0.06) \\
240 & 40 & 50 & 0.15(0.02) & 0.14(0.02) & 0.14(0.01) & 0.14(0.01) & 0.46(0.05) & 0.42(0.06) & 0.41(0.06) & 0.41(0.04) & 0.41(0.03) & 0.53(0.06) \\ \hline
60 & 10 & 100 & 0.63(0.06) & 0.63(0.07) & 0.63(0.05) & 0.63(0.05) & 0.36(0.07) & 1.13(0.12) & 1.13(0.13) & 1.13(0.1) & 1.13(0.09) & 0.48(0.09) \\
120 & 10 & 100 & 0.43(0.04) & 0.43(0.04) & 0.43(0.03) & 0.43(0.03) & 0.35(0.07) & 0.77(0.08) & 0.77(0.07) & 0.77(0.05) & 0.77(0.05) & 0.44(0.08) \\
240 & 10 & 100 & 0.3(0.03) & 0.3(0.03) & 0.3(0.02) & 0.3(0.02) & 0.34(0.07) & 0.54(0.05) & 0.53(0.05) & 0.54(0.03) & 0.54(0.03) & 0.41(0.08) \\ \hdashline
60 & 20 & 100 & 0.47(0.05) & 0.47(0.05) & 0.47(0.04) & 0.47(0.04) & 0.35(0.05) & 1.01(0.11) & 1.02(0.11) & 1.01(0.08) & 1.01(0.08) & 0.47(0.06) \\
120 & 20 & 100 & 0.32(0.03) & 0.32(0.03) & 0.32(0.02) & 0.32(0.02) & 0.34(0.05) & 0.68(0.07) & 0.68(0.07) & 0.68(0.05) & 0.68(0.05) & 0.41(0.05) \\
240 & 20 & 100 & 0.22(0.02) & 0.22(0.02) & 0.22(0.01) & 0.22(0.01) & 0.32(0.05) & 0.47(0.04) & 0.47(0.04) & 0.47(0.03) & 0.47(0.03) & 0.37(0.05) \\ \hdashline
60 & 40 & 100 & 0.29(0.03) & 0.29(0.03) & 0.29(0.02) & 0.29(0.02) & 0.34(0.04) & 0.77(0.1) & 0.77(0.1) & 0.77(0.07) & 0.77(0.07) & 0.47(0.04) \\
120 & 40 & 100 & 0.2(0.02) & 0.2(0.02) & 0.2(0.01) & 0.2(0.01) & 0.32(0.04) & 0.52(0.05) & 0.51(0.05) & 0.52(0.04) & 0.52(0.04) & 0.4(0.04) \\
240 & 40 & 100 & 0.14(0.01) & 0.14(0.01) & 0.14(0.01) & 0.14(0.01) & 0.32(0.03) & 0.35(0.03) & 0.36(0.03) & 0.35(0.02) & 0.35(0.02) & 0.35(0.04) \\ \hline
60 & 10 & 200 & 0.63(0.05) & 0.62(0.05) & 0.63(0.04) & 0.63(0.04) & 0.26(0.06) & 1.11(0.08) & 1.1(0.08) & 1.1(0.07) & 1.1(0.07) & 0.33(0.07) \\
120 & 10 & 200 & 0.43(0.03) & 0.43(0.03) & 0.43(0.02) & 0.43(0.02) & 0.25(0.05) & 0.77(0.05) & 0.76(0.05) & 0.77(0.04) & 0.77(0.04) & 0.31(0.06) \\
240 & 10 & 200 & 0.3(0.02) & 0.3(0.02) & 0.3(0.01) & 0.3(0.01) & 0.24(0.05) & 0.54(0.03) & 0.54(0.03) & 0.54(0.02) & 0.54(0.02) & 0.29(0.06) \\ \hdashline
60 & 20 & 200 & 0.47(0.04) & 0.47(0.04) & 0.47(0.03) & 0.47(0.03) & 0.25(0.03) & 0.99(0.07) & 0.98(0.07) & 0.98(0.06) & 0.98(0.06) & 0.34(0.05) \\
120 & 20 & 200 & 0.32(0.02) & 0.32(0.02) & 0.32(0.02) & 0.32(0.02) & 0.24(0.04) & 0.68(0.05) & 0.67(0.04) & 0.67(0.04) & 0.67(0.03) & 0.29(0.04) \\
240 & 20 & 200 & 0.22(0.01) & 0.22(0.01) & 0.22(0.01) & 0.22(0.01) & 0.23(0.03) & 0.47(0.03) & 0.47(0.03) & 0.47(0.02) & 0.47(0.02) & 0.26(0.04) \\ \hdashline
60 & 40 & 200 & 0.29(0.03) & 0.29(0.02) & 0.29(0.02) & 0.29(0.02) & 0.24(0.03) & 0.73(0.06) & 0.73(0.05) & 0.73(0.05) & 0.73(0.05) & 0.33(0.04) \\
120 & 40 & 200 & 0.2(0.01) & 0.2(0.01) & 0.2(0.01) & 0.2(0.01) & 0.23(0.02) & 0.5(0.03) & 0.5(0.03) & 0.5(0.03) & 0.5(0.03) & 0.28(0.03) \\
240 & 40 & 200 & 0.14(0.01) & 0.14(0.01) & 0.14(0.01) & 0.14(0.01) & 0.22(0.02) & 0.35(0.02) & 0.35(0.02) & 0.35(0.01) & 0.35(0.01) & 0.25(0.03) \\ \hline
60 & 10 & 400 & 0.61(0.04) & 0.61(0.04) & 0.61(0.04) & 0.61(0.04) & 0.18(0.04) & 1.08(0.07) & 1.08(0.07) & 1.08(0.06) & 1.08(0.06) & 0.24(0.05) \\
120 & 10 & 400 & 0.42(0.02) & 0.42(0.02) & 0.42(0.02) & 0.42(0.02) & 0.17(0.04) & 0.75(0.04) & 0.75(0.04) & 0.75(0.03) & 0.75(0.03) & 0.22(0.05) \\
240 & 10 & 400 & 0.3(0.01) & 0.3(0.01) & 0.3(0.01) & 0.3(0.01) & 0.17(0.04) & 0.52(0.02) & 0.53(0.02) & 0.53(0.02) & 0.53(0.02) & 0.2(0.04) \\ \hdashline
60 & 20 & 400 & 0.46(0.03) & 0.46(0.03) & 0.46(0.03) & 0.46(0.03) & 0.18(0.03) & 0.95(0.05) & 0.95(0.06) & 0.95(0.05) & 0.95(0.05) & 0.24(0.04) \\
120 & 20 & 400 & 0.31(0.02) & 0.31(0.02) & 0.31(0.01) & 0.31(0.01) & 0.17(0.02) & 0.65(0.04) & 0.65(0.03) & 0.65(0.03) & 0.65(0.03) & 0.2(0.03) \\
240 & 20 & 400 & 0.22(0.01) & 0.22(0.01) & 0.22(0.01) & 0.22(0.01) & 0.16(0.02) & 0.46(0.02) & 0.46(0.02) & 0.46(0.01) & 0.46(0.01) & 0.18(0.03) \\ \hdashline
60 & 40 & 400 & 0.29(0.02) & 0.29(0.02) & 0.29(0.02) & 0.29(0.02) & 0.17(0.02) & 0.7(0.04) & 0.7(0.05) & 0.7(0.04) & 0.7(0.04) & 0.24(0.02) \\
120 & 40 & 400 & 0.19(0.01) & 0.19(0.01) & 0.19(0.01) & 0.19(0.01) & 0.16(0.02) & 0.49(0.02) & 0.48(0.02) & 0.48(0.02) & 0.48(0.02) & 0.2(0.02) \\
240 & 40 & 400 & 0.13(0.01) & 0.13(0.01) & 0.13(0) & 0.13(0) & 0.16(0.02) & 0.34(0.02) & 0.34(0.01) & 0.34(0.01) & 0.34(0.01) & 0.18(0.02) \\ \hline
\end{tabular}%
}
\end{table}

\begin{table}[htpb!]
\centering
\caption{Mean and standard deviations (in parentheses) of the mean squared prediction errors (MSPE).}
\label{table:STprediction}
\resizebox{\textwidth}{!}{%
\begin{tabular}{ccc|c|cc|cc}
\hline
 &  &  & Spatial & \multicolumn{2}{c|}{Temporal MAR(1)} & \multicolumn{2}{c}{Temporal VAR(1)} \\ \hline
T & p & n & $MSPE(\hat{\by}_t(\bs_0)))$ & $MSPE(\hat{\by}_{t+1}(\bs)))$ & $MSPE(\hat{\by}_{t+2}(\bs)))$ & $MSPE(\hat{\by}_{t+1}(\bs)))$ & $MSPE(\hat{\by}_{t+2}(\bs)))$ \\ \hline
60 & 10 & 50 & 0.486(0.089) & 1.716(1.064) & 1.823(1.201) & 1.825(1.075) & 2.019(1.257) \\
120 & 10 & 50 & 0.471(0.06) & 1.658(1.121) & 1.634(1.116) & 1.705(1.133) & 1.732(1.144) \\
240 & 10 & 50 & 0.47(0.041) & 1.78(1.079) & 1.588(1.244) & 1.802(1.076) & 1.624(1.229) \\ \hdashline
60 & 20 & 50 & 0.424(0.069) & 1.592(1.004) & 1.657(1.033) & 1.69(1.032) & 1.819(1.061) \\
120 & 20 & 50 & 0.424(0.048) & 1.535(0.972) & 1.547(1.111) & 1.575(0.983) & 1.634(1.128) \\
240 & 20 & 50 & 0.419(0.036) & 1.619(0.985) & 1.426(1.05) & 1.64(0.988) & 1.463(1.047) \\ \hdashline
60 & 40 & 50 & 0.537(0.085) & 2.001(1.237) & 2.101(1.353) & 2.13(1.276) & 2.308(1.39) \\
120 & 40 & 50 & 0.534(0.055) & 2.006(1.345) & 1.94(1.286) & 2.065(1.36) & 2.051(1.296) \\
240 & 40 & 50 & 0.53(0.037) & 2.141(1.434) & 1.834(1.237) & 2.162(1.432) & 1.877(1.23) \\ \hline
60 & 10 & 100 & 0.067(0.009) & 1.597(0.966) & 1.647(1.006) & 1.685(0.969) & 1.82(1.03) \\
120 & 10 & 100 & 0.066(0.006) & 1.564(0.984) & 1.502(0.95) & 1.608(0.997) & 1.593(0.973) \\
240 & 10 & 100 & 0.065(0.004) & 1.631(0.92) & 1.476(1.02) & 1.65(0.915) & 1.514(1.015) \\ \hdashline
60 & 20 & 100 & 0.058(0.008) & 1.466(0.876) & 1.508(0.901) & 1.557(0.891) & 1.663(0.926) \\
120 & 20 & 100 & 0.058(0.005) & 1.45(0.883) & 1.403(0.915) & 1.489(0.891) & 1.478(0.922) \\
240 & 20 & 100 & 0.058(0.004) & 1.491(0.856) & 1.317(0.864) & 1.51(0.854) & 1.353(0.859) \\ \hdashline
60 & 40 & 100 & 0.072(0.01) & 1.845(1.075) & 1.893(1.105) & 1.975(1.113) & 2.085(1.126) \\
120 & 40 & 100 & 0.072(0.006) & 1.889(1.229) & 1.765(1.076) & 1.939(1.247) & 1.859(1.077) \\
240 & 40 & 100 & 0.072(0.005) & 1.961(1.223) & 1.707(1.074) & 1.984(1.22) & 1.754(1.068) \\ \hline
60 & 10 & 200 & 0.015(0.002) & 1.542(0.922) & 1.597(0.972) & 1.629(0.921) & 1.766(1) \\
120 & 10 & 200 & 0.015(0.001) & 1.515(0.976) & 1.454(0.913) & 1.557(0.982) & 1.538(0.934) \\
240 & 10 & 200 & 0.015(0.001) & 1.599(0.915) & 1.42(0.988) & 1.619(0.912) & 1.458(0.988) \\ \hdashline
60 & 20 & 200 & 0.013(0.002) & 1.419(0.86) & 1.461(0.88) & 1.51(0.88) & 1.61(0.897) \\
120 & 20 & 200 & 0.013(0.001) & 1.401(0.853) & 1.358(0.88) & 1.44(0.861) & 1.429(0.883) \\
240 & 20 & 200 & 0.013(0.001) & 1.464(0.859) & 1.276(0.84) & 1.481(0.86) & 1.308(0.838) \\ \hdashline
60 & 40 & 200 & 0.015(0.002) & 1.786(1.04) & 1.836(1.099) & 1.906(1.066) & 2.02(1.122) \\
120 & 40 & 200 & 0.015(0.001) & 1.828(1.211) & 1.714(1.042) & 1.875(1.22) & 1.808(1.049) \\
240 & 40 & 200 & 0.015(0.001) & 1.92(1.214) & 1.652(1.031) & 1.941(1.213) & 1.698(1.027) \\ \hline
60 & 10 & 400 & 0.014(0.002) & 1.63(0.965) & 1.714(1.033) & 1.727(0.965) & 1.893(1.059) \\
120 & 10 & 400 & 0.014(0.001) & 1.63(1.058) & 1.556(0.975) & 1.676(1.069) & 1.647(1.009) \\
240 & 10 & 400 & 0.014(0.001) & 1.711(0.985) & 1.527(1.077) & 1.728(0.983) & 1.568(1.075) \\ \hdashline
60 & 20 & 400 & 0.012(0.002) & 1.511(0.914) & 1.561(0.926) & 1.611(0.936) & 1.719(0.949) \\
120 & 20 & 400 & 0.012(0.001) & 1.502(0.923) & 1.452(0.934) & 1.543(0.931) & 1.534(0.945) \\
240 & 20 & 400 & 0.012(0.001) & 1.569(0.929) & 1.373(0.915) & 1.589(0.931) & 1.407(0.912) \\ \hdashline
60 & 40 & 400 & 0.015(0.002) & 1.907(1.108) & 1.964(1.166) & 2.033(1.14) & 2.159(1.181) \\
120 & 40 & 400 & 0.015(0.001) & 1.967(1.319) & 1.831(1.107) & 2.021(1.334) & 1.937(1.117) \\
240 & 40 & 400 & 0.015(0.001) & 2.062(1.314) & 1.775(1.118) & 2.086(1.31) & 1.823(1.111) \\ \hline
\end{tabular}%
}
\end{table}